\documentclass{article}
\usepackage{arxiv}

\usepackage[utf8]{inputenc} %
\usepackage[T1]{fontenc}    %
\usepackage{hyperref}       %
\usepackage{url}            %
\usepackage{booktabs}       %
\usepackage{amsfonts}       %
\usepackage{nicefrac}       %
\usepackage{microtype}      %
\usepackage{lipsum}		%
\usepackage{graphicx}
\usepackage{natbib}
\usepackage{doi}

\usepackage{lipsum}
\usepackage{amsfonts}
\usepackage{amsmath}
\usepackage{amsthm}
\usepackage{amssymb}
\usepackage{graphicx}
\usepackage{epstopdf}
\usepackage{soul}
\usepackage{caption}
\usepackage{subcaption}
\usepackage{algorithm}
\usepackage{algpseudocode}
\usepackage{enumitem}
\usepackage{comment}
\ifpdf
  \DeclareGraphicsExtensions{.eps,.pdf,.png,.jpg}
\else
  \DeclareGraphicsExtensions{.eps}
\fi

\usepackage[dvipsnames,x11names]{xcolor}

\usepackage{booktabs}
\usepackage{multirow}
\usepackage{annotate-equations}
\usepackage{tikz}
\usepackage{multirow}
\usetikzlibrary{backgrounds}
\usetikzlibrary{arrows,shapes}
\usetikzlibrary{tikzmark} %

\definecolor{darkred}{rgb}{.7,0,0}

\definecolor{darkgreen}{rgb}{.15,.55,0}

\definecolor{darkblue}{rgb}{0,0,0.7}

\newcommand{\ExpLoss}{\overline{\mathcal{L}}}

\newcommand{\Var}{\mathbb{V}\text{ar}}
\newcommand{\Cov}{\mathbb{C}\text{ov}}

\newcommand{\CA}{\mathcal{A}}

\newcommand{\CL}{\mathcal{L}}

\newcommand{\EE}{\mathbb{E}}

\newcommand{\RR}{\mathbb{R}}

\colorlet{colorp}{NavyBlue}
\colorlet{colorT}{WildStrawberry}
\colorlet{colork}{OliveGreen}
\colorlet{colorM}{RoyalPurple}
\colorlet{colorNb}{Plum!80}
\colorlet{colorIs}{black}
\colorlet{customrgbgreen}{darkgreen}
\colorlet{colorY}{Dandelion}


\usepackage{enumitem}
\setlist[enumerate]{leftmargin=.5in}
\setlist[itemize]{leftmargin=.5in}

\newtheorem{theorem}{Theorem}[section]
\newtheorem{lemma}[theorem]{Lemma}
\title{Bayesian Covariance Uncertainty for Adaptive Pilot-Sampling Termination in Multi-fidelity Uncertainty Quantification}
\date{}

\newif\ifuniqueAffiliation
\uniqueAffiliationtrue

\ifuniqueAffiliation %
\author{ \href{https://orcid.org/0000-0001-6764-9275}{\includegraphics[scale=0.06]{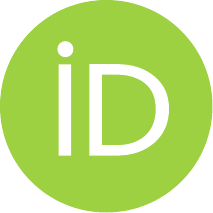}\hspace{1mm}Thomas E. Coons}\\
	Department of Mechanical Engineering\\
	University of Michigan\\
	Ann Arbor, MI 48109 \\
	\texttt{tcoons@umich.edu} \\
	\And
	\href{https://orcid.org/0000-0003-3379-7514}{\includegraphics[scale=0.06]{orcid.pdf}\hspace{1mm}Aniket Jivani} \\
        Department of Mechanical Engineering\\
	University of Michigan\\
	Ann Arbor, MI 48109 \\
    \texttt{ajivani@umich.edu}
    \And
	\href{https://orcid.org/0000-0001-6544-2764}{\includegraphics[scale=0.06]{orcid.pdf}\hspace{1mm}Xun Huan} \\
    Department of Mechanical Engineering\\
	University of Michigan\\
	Ann Arbor, MI 48109 \\
    \texttt{xhuan@umich.edu}
}
\fi
\usepackage{amsopn}
\DeclareMathOperator{\diag}{diag}
\renewcommand{\shorttitle}{Bayesian Covariance Uncertainty for Adaptive Pilot-Sampling Termination}

\hypersetup{
pdftitle={Bayesian UQ for MFUQ},
pdfauthor={Thomas E. Coons, Aniket Jivani, Xun Huan},
}

\begin{document}
\maketitle
\begin{abstract}
    Monte Carlo integration becomes prohibitively expensive when each sample requires a high-fidelity model evaluation. 
    Multi-fidelity uncertainty quantification methods mitigate this by combining estimators from high- and low-fidelity models, preserving unbiasedness while reducing variance under a fixed budget. 
    Constructing such estimators optimally requires the model-output covariance matrix, typically estimated from pilot samples. Too few pilot samples lead to inaccurate covariance estimates and suboptimal estimators, while too many consume budget that could be used for final estimation. 
    We propose a Bayesian framework to quantify covariance uncertainty from pilot samples, incorporating prior knowledge and enabling probabilistic assessments of estimator performance.
    A central component is a flexible $\gamma$-Gaussian prior that ensures computational tractability and supports efficient posterior projection under additional pilot samples. These tools enable adaptive pilot-sampling termination via an interpretable loss criterion that decomposes variance inefficiency into accuracy and cost components.
    While demonstrated here in the context of approximate control variates (ACV), the framework generalizes to other multi-fidelity estimators. We validate the approach on a monomial benchmark and a PDE-based Darcy flow problem.
    Across these tests, our adaptive method demonstrates its value for multi-fidelity estimation under limited pilot budgets and expensive models,
    achieving variance reduction comparable to baseline estimators with oracle covariance.
\end{abstract}


\section{Introduction}

Estimating the expected value of a function is a central task in uncertainty quantification (UQ). While alternatives such as quadrature and Quasi Monte Carlo exist, Monte Carlo (MC) integration remains the standard approach, especially in high-dimensional settings. However, MC can become prohibitively expensive when each sample requires a high-fidelity model evaluation (e.g., solving a partial differential equation (PDE) on fine discretizations). Replacing the high-fidelity model with a surrogate model can reduce cost but typically introduces bias. Multi-fidelity UQ methods provide a more powerful alternative: by intelligently {combining} estimators from high- and low-fidelity models, they preserve unbiasedness while reducing variance under a fixed computational budget.

A variety of multi-fidelity estimators have proven effective, including multilevel Monte Carlo (MLMC)~\cite{Giles2008}, multi-fidelity Monte Carlo (MFMC)~\cite{peherstorfer_optimal_2016}, multilievel best linear unbiased estimators (MLBLUE) \cite{schaden_multilevel_2020}, approximate control variates (ACV) \cite{Gorodetsky2020, bomarito_optimization_2022}, 
and 
grouped ACV (GACV) 
\cite{gorodetsky_grouped_2024}.
All of these sampling-based estimators are parameterized by weights and a sample allocation,
obtained by solving an optimization problem for a given budget. 
Crucially, this optimization requires the covariance matrix of the model outputs, which is generally unknown \textit{a priori}. 

The standard remedy is \textit{pilot sampling}: evaluate all models on an independent set of inputs, estimate the covariance empirically, then optimize the weights and sample allocation. Pilot sampling, however, introduces a fundamental trade-off: 
\begin{itemize}
\item Taking too few samples results in inaccurate covariance estimates, which in turn leads to suboptimal weights and sample allocation.
\item Taking too many samples produces accurate covariance estimates, but leaves insufficient budget for the final estimation.
\end{itemize}
For example, the left panel of Figure~\ref{fig:pilot_study} illustrates an optimal region of estimator variance versus pilot budget (out of a fixed total budget equivalent to 200 pilot samples). In addition, the predicted estimator error using too few pilot samples can underestimate the true estimator variance by orders of magnitude, creating overconfidence in estimator quality. For example, the right panel of Figure~\ref{fig:pilot_study} highlights this effect.
Despite their importance, these issues are rarely examined, and no general framework exists for quantifying covariance uncertainty for multi-fidelity UQ. 
To our knowledge, the only existing published work that proposed a solution to this trade-off is~\cite{xu_bandit-learning_2022}, which relies on an asymptotic analysis of a non-optimal estimator (see Section~\ref{sec:AETC}). 
Alternative strategies include exchanging pilot sampling for ensemble estimation~\cite{pham_ensemble_2022} while ignoring covariance uncertainty, or fixing weights at suboptimal values that avoids dependence on covariance information~\cite{chakroborty_covariance-free_2024}.
In contrast, we directly incorporate covariance uncertainty, enabling a flexible and general approach to adaptive pilot-sampling termination in finite-sample regimes.

\begin{figure}[htbp]
\centering
  \centering
  \includegraphics[width=.48\textwidth]{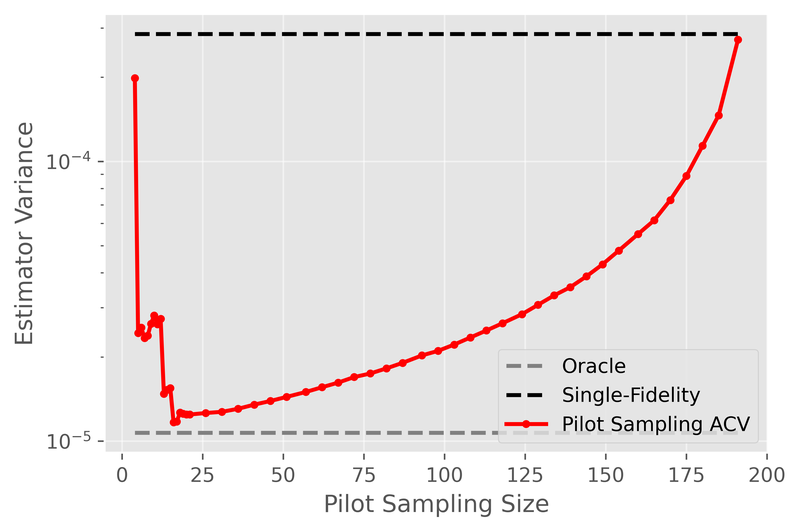}
  \includegraphics[width=.48\textwidth]{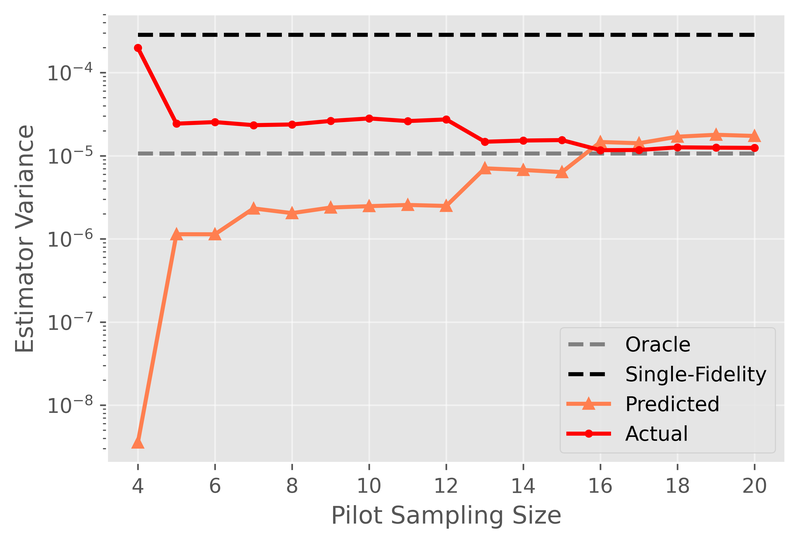}
\caption{
    Example (see Section~\ref{sec:monomial} for details) illustrating ACV estimator variance (left) and variance-prediction error (right) across different pilot-sample sizes. The black dashed line indicates the variance of the single-fidelity MC estimator, while the gray dashed line indicates the ACV variance under exact covariance knowledge (oracle case).
}
\label{fig:pilot_study}
\end{figure}

We address two key questions: 
\begin{enumerate}
\item How can we quantify uncertainty in covariance estimates obtained from pilot samples? 
\item How can we adaptively decide when to terminate pilot sampling under a given budget?
\end{enumerate}
Our approach is Bayesian. We treat the model-output covariance as an uncertain parameter and perform Bayesian inference conditioned on pilot data to obtain a posterior over covariance matrices. 
At each pilot iteration, we use the posterior to evaluate an expected-loss criterion that measures the inefficiency from additional pilot samples. We terminate pilot sampling once further pilot samples are projected to decrease efficiency. After termination, the posterior also provides a probabilistic assessment of estimator error, in contrast to point predictions common in prior work.

The main contributions of this work are:
\begin{itemize}
\setlength
    \item We propose a Bayesian approach to quantify the uncertainty in covariance estimation from pilot samples, enabling the incorporation of prior knowledge of the model covariance and a probabilistic assessment of estimator performance.
    \item We develop a $\gamma$-Gaussian prior for correlation matrices (extending the initial formulation of \cite{archakov_new_2021}) and design efficient methods to project posteriors under additional pilot samples, applicable to both this prior and the inverse Wishart distribution. This capability is essential for incorporating covariance uncertainty into adaptive pilot-sampling termination. 
    \item We introduce an interpretable loss function that decomposes into accuracy and cost components, %
    and use it to guide adaptive pilot-sampling termination.
    \item We validate the method on a monomial benchmark problem and a PDE-based Darcy flow application.
\end{itemize}
These contributions provide a principled and general framework for adaptive pilot sampling across a wide range of multi-fidelity estimation settings.
While this paper focuses on ACV for concreteness, our developments can also be applied to other estimators such as MLBLUE and GACV.
Our code is available at: \url{https://github.com/tcoonsUM/bayes_mfuq}.

The remainder of this paper is organized as follows. Section~\ref{sec:background} reviews ACV estimators and related work. Section~\ref{sec:bayes_inf} introduces Bayesian inference for covariance matrices, including the new $\gamma$-Gaussian prior. Section~\ref{sec:adaptive_pilot_sampling} presents the expected-loss criterion and the adaptive termination algorithm. Section~\ref{sec:demos} demonstrates the method on benchmark and PDE-based problems. Section~\ref{sec:conclusions}
 concludes.

\section{Background and related work}\label{sec:background}
In this section, we provide background on multi-fidelity estimation using ACV estimators and review related work on balancing pilot sampling and estimation costs.

\subsection{Approximate control variate estimators}\label{sec:ACVs}

Consider a mapping $\boldsymbol{Q}=f_0(\boldsymbol{Z})$, where $\boldsymbol{Z} \in \mathbb{R}^{n_{z}}$ is a random vector and $\boldsymbol{Q} \in \mathbb{R}$ is a scalar random variable.\footnote{We use boldface to denote random quantities, and non-bold symbols to denote their realizations and other non-random quantities.} The expected value $\mathbb{E}[\boldsymbol{Q}]$ is typically estimated using the standard MC estimator:
\begin{align}\label{eqn:mc}
    \mathbb{E}\left[\boldsymbol{Q}\right] 
    \approx \boldsymbol{\hat{Q}}_0(\boldsymbol{z}) = \frac{1}{N} \sum^{N}_{j=1}f_0(\boldsymbol{Z}^{(j)}),
\end{align}
where $\boldsymbol{z}=[\boldsymbol{Z}^{(1)},\ldots,\boldsymbol{Z}^{(N)}]$ is the set of independent and identically distributed (i.i.d.) random samples following $p(Z)$.\footnote{Formally, the probability density is written as $p_{\boldsymbol{Z}}(\boldsymbol{Z}=Z)$, but for brevity we write $p(Z)$ or $p_{\boldsymbol{Z}}$.
} This estimator is unbiased, and its error is governed entirely by its variance: $\Var[\boldsymbol{\hat{Q}}_0] = N^{-1}\Var[\boldsymbol{Q}]$. When $f_0$ is computationally expensive, the available budget may not allow a sufficiently large $N$ to reduce this variance to acceptable levels.

To mitigate this difficulty, multi-fidelity methods leverage an ensemble of lower-fidelity models $f_m(\boldsymbol{Z})$ that are correlated with $f_0(\boldsymbol{Z})$ but less expensive to evaluate. The ACV method \cite{Gorodetsky2020, bomarito_optimization_2022} introduces $M-1$ such auxiliary random variables $\boldsymbol{Q}_{m} =f_{m}(\boldsymbol{Z})$ 
and constructs a combined estimator:
\begin{align}
        \boldsymbol{\Tilde{Q}}(\boldsymbol{z};\alpha,\CA) &= 
        \boldsymbol{\hat{Q}}_{0}(\boldsymbol{z}_{0})+\sum^{M-1}_{m=1}\alpha_{m}\left( \boldsymbol{\hat{Q}}_{m}(\boldsymbol{z}^{\ast}_{m})-\boldsymbol{\hat{Q}}_{m}(\boldsymbol{z}_{m}) \right),     \label{eqn:ACV-Formula}
\end{align}
where $\alpha = \left[ \alpha_{1}, \ldots, \alpha_{M-1} \right]$ are control variate weights, 
and the input samples $\boldsymbol{z}$ are partitioned into subsets $\boldsymbol{z}_{0}$, $\boldsymbol{z}_{m}$, and $\boldsymbol{z}^{\ast}_{m}$ according to a sample allocation $\CA$. 
Like standard MC, the ACV estimator is unbiased, so its error is entirely determined by its variance.

Using vectorized notation, we define $\boldsymbol{\Delta}=\left[ \boldsymbol{\Delta}_{1}(\boldsymbol{z}_{1}^{\ast},\boldsymbol{z}_{1}), \ldots,  \boldsymbol{\Delta}_{M}(\boldsymbol{z}_{M}^{\ast},\boldsymbol{z}_{M})\right]$, where each difference term is $\boldsymbol{\Delta}_{m}(\boldsymbol{z}_{m}^{\ast},\boldsymbol{z}_{m})=\boldsymbol{\hat{Q}}_{m}(\boldsymbol{z}^{\ast}_{m})-\boldsymbol{\hat{Q}}_{m}(\boldsymbol{z}_{m})$. Using this notation, the ACV estimator can be written compactly as:
\begin{align}\label{eqn:vectorized-acv}
    \boldsymbol{\Tilde{Q}}(\boldsymbol{z};\alpha,\CA) = \boldsymbol{\hat{Q}}_{0} + \alpha^{\top}\boldsymbol{\Delta}.
\end{align}
The variance of the ACV estimator is:
\begin{align}\label{eqn:acv-variance}
    \Var[\boldsymbol{\Tilde{Q}}(\boldsymbol{z};\alpha,\CA)] = \Var[\boldsymbol{\hat{Q}}_{0}] - \alpha^{\top}\Cov[\boldsymbol{\Delta}, \boldsymbol{\Delta}]\alpha + 2 \alpha^{\top} \Cov[\boldsymbol{\Delta}, \boldsymbol{\hat{Q}}_{0}].
\end{align}
This variance depends on the covariances among model outputs, the control variate weights $\alpha$, and the sample partitioning strategy $\CA$, which affects how each $\boldsymbol{\Delta}$ term covaries with the others and with $\boldsymbol{\hat{Q}}_{0}$. If the exact model output covariance is known, then the optimal weights for a given $\CA$ can be obtained by minimizing the variance expression in \eqref{eqn:acv-variance}, yielding:
\begin{align}\label{eqn:alpha-star-acv}
    \alpha^{\ast}(\CA)=-\Cov[\boldsymbol{\Delta}, \boldsymbol{\Delta}]^{-1}\Cov[\boldsymbol{\Delta}, \boldsymbol{\hat{Q}}_{0}].
\end{align}
Substituting these optimal weights into \eqref{eqn:acv-variance} gives the corresponding minimized estimator variance~\cite{Gorodetsky2020}:
\begin{align}\label{eqn:acv-variance-opt}
    \Var[\boldsymbol{\Tilde{Q}}^{\alpha^{\ast}}](\CA) = \Var[\boldsymbol{\hat{Q}}_{0}] - \Cov[\boldsymbol{\Delta}, \boldsymbol{\hat{Q}}_{0}]^{\top} \Cov[\boldsymbol{\Delta}, \boldsymbol{\Delta}]^{-1}\Cov[\boldsymbol{\Delta}, \boldsymbol{\hat{Q}}_{0}].
\end{align}

The sample allocation $\CA$ determines how the computational workload is distributed across the models. For example, the well-known MLMC \cite{giles_multilevel_2015}  and MFMC \cite{peherstorfer_optimal_2016} methods can be shown to be special cases of ACV, distinguished by their restricted families of allowable allocations $\CA$. Solving for the optimal sample partitioning strategy $\CA^{\ast}$ is generally more involved than solving for the optimal weights $\alpha^{\ast}$, but several efficient tools have been developed to approximately solve the following optimization problem \cite{bomarito_multi_2020, jakeman_pyapprox_2023}:
\begin{align}
    \min_{\mathcal{A}\in\mathbb{A}} &\quad \Var[\boldsymbol{\Tilde{Q}}(\boldsymbol{z};\alpha^{\ast},\mathcal{A})] \label{eqn:mxmcpy}\\
    \mathrm{subject \,\, to} &\quad \mathcal{W}(w,\mathcal{A})\leq 
    B,\label{eqn:mxmcpy_constraint}
\end{align}
where $\mathbb{A}$ is a predefined set of allowable sample allocations, $B$ is the total computational budget, $w = [w_0, \ldots, w_M]$ are the per-sample model costs, and $\mathcal{W}(w,\mathcal{A})$ denotes the total cost of the estimator. Practically, the total budget $B$ must be sufficiently large to allow for at least one high-fidelity model evaluation under the optimal allocation $\CA^\ast$; we will always assume this is the case in this paper.

Importantly, solving both \eqref{eqn:alpha-star-acv} for $\alpha^{\ast}$ and \eqref{eqn:mxmcpy} for $\CA^{\ast}$ requires knowledge of the model-output covariance matrix: 
\begin{align}\label{eqn:cov_defn}
    \Sigma = \Cov\left[ f_0(\boldsymbol{Z}),\ldots,f_{M-1}(\boldsymbol{Z})
    \right].
\end{align}
While the relationship between the weights $\alpha$ and the estimator variance is explicit in \eqref{eqn:acv-variance}, 
the dependence of the estimator variance on $\CA$ and $\Sigma$ arises via the terms:
\begin{align}
    \Cov[\boldsymbol{\Delta}, \boldsymbol{\Delta}] &= G(\CA) \odot \Sigma, \label{eqn:cov-delta-delta} \\
    \Cov[\boldsymbol{\Delta}, \boldsymbol{\hat{Q}}_{0}] &= \mathrm{diag}(G(\CA)) \odot \mathrm{col}_0( \Sigma ), \label{eqn:cov-delta-Q}
\end{align}
where $G \in \RR^{M\times M}$ is a matrix that depends on $\CA$ 
(see \cite{bomarito_multi_2020}), $\odot$ denotes the Hadamard (elementwise) product, and $\mathrm{col}_{0}( \Sigma )$ is the first column of $\Sigma$, representing the covariances between the high-fidelity model and each low-fidelity model.

\sloppy Although model costs $w$ are often known in advance, the covariance matrix $\Sigma$ is typically unknown \textit{a priori}. A common approach is to estimate it using $N_{\mathrm{pilot}}$ independent pilot evaluations via the unbiased sample covariance formula:
\begin{align}\label{eqn:cov}
    \hat{\Sigma} = \frac{1}{N_{\mathrm{pilot}}-1}\sum_{j=1}^{N_{\mathrm{pilot}}}\left( f^{(j)} - \overline{f} \right) \left( f^{(j)} - \overline{f} \right)^{\top},
\end{align}
where 
$f^{(j)}=[ f_{0}(z^{(j)}),\ldots,f_{M-1}(z^{(j)}) ]^{\top}$ and $\overline{f}=[ \overline{f}_{0},\ldots,\overline{f}_{M-1} ]^{\top}$ with $\overline{f}_{m}=\frac{1}{N_{\mathrm{pilot}}}\sum_{j=1}^{N_{\mathrm{pilot}}}f_{m}(z^{(j)})$. 
Choosing too few pilot samples can result in inaccurate covariance estimates, leading to suboptimal weights or allocations and thus poor estimator performance. On the other hand, excessive pilot sampling may exhaust the computational budget, leaving insufficient resources for the final ACV estimator. 
While this trade-off has been highlighted empirically in prior studies~\cite{peherstorfer_optimal_2016, Gorodetsky2020}, there remains a need for a principled framework to quantify and balance the trade-off
between pilot sampling effort and estimator performance under covariance uncertainty.

\subsection{Related work: Adaptive Explore-Then-Commit (AETC) Linear Regression Monte Carlo (LRMC) estimator}\label{sec:AETC}

When multi-fidelity estimators rely on approximate covariance estimated from pilot data, the exact mean-squared errors (MSE) of the estimators are generally intractable. To address this, \cite{xu_bandit-learning_2022} introduced a new estimator called the Linear Regression Monte Carlo (LRMC) estimator. Unlike ACV, which constructs an estimator using control variate weights applied to MC differences, LRMC estimates the high-fidelity mean via a regression model over low-fidelity models, assuming a linear relationship between them.

Under this assumption, the total estimator MSE was decomposed into two components: (1) a bias term due to limited pilot samples (exploration), and (2) a variance term due to limited budget for evaluating the LRMC estimator (exploitation). This decomposition enables empirical estimation of the MSE trade-off and serves as the basis for adaptively terminating pilot sampling. Although the exact conditional MSE---conditioned on pilot data---remains intractable, the authors derived an asymptotic form that converges to the true MSE as $N_{\mathrm{pilot}}\rightarrow \infty$. 

To use this formulation in practice, the authors proposed an Adaptive Explore-Then-Commit (AETC) algorithm. This algorithm iteratively acquires pilot samples, uses them to empirically estimate the asymptotic conditional MSE, and terminates pilot sampling once this estimate begins to increase.

While this represents the first systematic method for balancing pilot sampling with evaluation budget, the approach has several limitations. 
\begin{itemize}
\item \textbf{Suboptimal sample allocation.} 
Unlike other multi-fidelity estimators such as ACV, MLBLUE, or GACV---which explicitly optimize sample allocation across model fidelities to minimize error---the LRMC estimator uses a uniform allocation across all low-fidelity models. That is, each low-fidelity model is sampled equally, regardless of its relative cost or predictive utility. 
\item \textbf{Dependence on asymptotic assumptions.} The asymptotic MSE formulation does not hold in the small-sample regime. In such settings, the sample covariance estimates used to construct the LRMC estimator may be highly inaccurate and noisy. 
\end{itemize}
These issues are especially pronounced when the total computational budget is small relative to model evaluation costs, as may be the case when high-fidelity models are expensive to evaluate. 

\vspace{1em}
Our proposed method addresses these shortcomings in several ways. First, it is agnostic to the choice of multi-fidelity estimator, and it is here applied to ACV estimators that leverage optimal sample allocations, which typically outperform uniform allocations. 
Second, our approach is well suited to the small-sample regime, as it directly incorporates Bayesian UQ over the covariance matrix. 
Finally, the use of prior distributions to inform covariance estimation enables regularization and improved estimation accuracy during pilot sampling, particularly when data are limited. 

In essence, the method presented in this paper can be viewed as a Bayesian generalization of AETC---one that does not rely on the restrictive structure of LRMC and more flexibly adapts to practical budget-constrained multi-fidelity estimation settings.

\section{Bayesian inference of covariance matrices}\label{sec:bayes_inf}

Bayesian statistical inference
provides a principled framework for quantifying uncertainty in the covariance matrix $\Sigma$ from \eqref{eqn:cov_defn}, especially when only a limited number of pilot samples are available. 
This uncertainty-aware formulation enables principled decision-making about whether additional pilot samples are likely to improve estimator performance and guides pilot-sampling termination decisions.

We treat $\boldsymbol{\Sigma} \in \RR^{M \times M}$, a symmetric and positive semidefinite (PSD) covariance matrix, as a random matrix of uncertain parameters, and apply Bayes' rule to obtain the posterior distribution:
\begin{align}
    p(\Sigma|D) = \frac{p(D | \Sigma) \, p(\Sigma)}{p(D)},
    \label{eqn: bayes_rule_v1_cov}
\end{align}
where $D =\{f(z^{(j)})\}_{j=1}^{N_{\mathrm{pilot}}}$ denotes the set of pilot-sample observations.
The prior distribution $p(\Sigma)$ encodes our belief about the covariance structure before any data are collected,
while the likelihood $p(D|\Sigma)$ models the probability of observations data $D$ given a particular covariance matrix $\Sigma$. 
Specific choices of prior and likelihood will be discussed in Sections \ref{ss: inv_wis} and \ref{ss: bij_param}.
The resulting posterior $p(\Sigma | D)$ provides a full probabilistic characterization of uncertainty in the covariance matrix after incorporating pilot observations, and can be used to guide pilot-sampling termination decisions. 

In practice, however, exact sampling from the posterior is typically intractable, except in special cases involving conjugate priors. Conjugacy arises only under specific combinations of prior and likelihood distributions. 
In more general settings, posterior inference must rely on methods such as Markov chain Monte Carlo (MCMC), which often require many evaluations of the likelihood (and potentially its derivatives), and can therefore be computationally expensive, particularly for high-dimensional or complex models. 
To avoid these costs, our approaches below leverage conjugate priors and are designed to sidestep MCMC.

\subsection{Inverse Wishart prior} \label{ss: inv_wis}

First introduced in \cite{wishart_generalised_1928}, the inverse Wishart (IW) prior is a widely used distribution for Bayesian inference over covariance matrices due to its conjugacy with multivariate Gaussian likelihoods. The IW prior, denoted $\mathrm{IW}(\nu,H)$, has a density:
\begin{align}
    p^{(\mathrm{IW})}(\Sigma; H, \nu) = \frac{\lvert H\rvert^{\nu/2}}{2^{(\nu M)/2}\,\Gamma_p(\nu/2)}\lvert\Sigma\rvert^{-(\nu+M+1)/2}e^{-\frac{1}{2}\mathrm{tr}(H\Sigma^{-1})},
\end{align}
where $H$ is a positive definite scale matrix that controls the location (center) of the distribution over $\boldsymbol{\Sigma}$, and $\nu>M+1$ is a degrees-of-freedom parameter that determines how tightly the distribution concentrates around the center. The mean and mode of an IW distribution with parameters $H$ and $\nu$ are:
\begin{align}\label{eqn:iw_mean}
    \EE\left[ \boldsymbol{\Sigma} \right] = \frac{H}{\nu - M - 1}, \qquad
    \mathrm{Mode}\left[ \boldsymbol{\Sigma} \right] = \frac{H}{\nu + M + 1}.
\end{align} 
The parameter $\nu$ can be interpreted as the effective number of pseudo-samples contributed by the prior, reflecting the strength of prior belief. A common heuristic is to set $\nu$ according to one's level of prior confidence; for example, setting $\nu=M+2$ yields a minimally informative prior with finite mean, while larger values correspond to stronger prior confidence.
Once $\nu$ is set, $H$ can then be chosen to match a desired prior mean or mode using \eqref{eqn:iw_mean}.

A standard assumption is that the likelihood is multivariate Gaussian, i.e., $p(D|\Sigma) \sim \mathcal{N}(\mu, \Sigma)$. Under this Gaussian likelihood, the IW prior is conjugate, and the posterior distribution over $\boldsymbol{\Sigma}$ remains IW, with updated parameters:
\begin{align}
    H_D &= H + N_\mathrm{pilot} \times S_D, \\
    \nu_D &= \nu+N_\mathrm{pilot},
\end{align}
where $S_D=\sum_{j=1}^{N_\mathrm{pilot}}(f^{(j)}-\overline{f})(f^{(j)}-\overline{f})^{\top}$
is the sample scatter matrix.

While the IW distribution is analytically convenient and computationally efficient, it suffers from limited flexibility.
As noted in \cite{alvarez_bayesian_2016}, its parameterization couples the variances and correlations, making it difficult to encode structured prior knowledge. For instance, in a multi-fidelity setting, one may have stronger prior beliefs regarding low-fidelity model variances but remain uncertain about the high-fidelity model or inter-model correlations.
However, the single parameter $\nu$ jointly controls prior confidence over all covariance entries, preventing independent specification of beliefs. 
To address this limitation, several works have proposed more flexible alternatives, including over-parameterized extensions of the IW prior and re-parameterizations that treat variances and correlations separately \cite{omalley_domain_level_2008, barnard2000}.

In the next section, we introduce a prior that enables structured specification of prior beliefs across individual components of the covariance matrix.

\subsection{\texorpdfstring{$\boldsymbol{\gamma}$-Gaussian prior}{gamma-Gaussian prior}}\label{ss: bij_param}

We propose the $\gamma$-Gaussian prior
as another choice for setting up covariance inference that overcomes the inflexibility of the IW prior by using a bijective parameterization of the correlation matrix.
Unlike IW priors, which require multivariate Gaussian likelihood for conjugacy, the $\gamma$-Gaussian prior
places the Gaussian assumption on a set of correlation parameters rather than on the model outputs, a choice shown in \cite{archakov_new_2021} to be suitable in many cases. 
This avoids unrealistic Gaussian output assumptions, which can be particularly problematic for chaotic or unstable model ensembles, and offers greater flexibility in encoding prior beliefs. 

\subsubsection{Bijective parameterizations}
We begin by reviewing existing parameterizations of the covariance or correlation matrix. Since any covariance matrix can be decomposed into its standard deviations and a correlation matrix, we seek a parameterization of either object while satisfying the following criteria:
\begin{itemize}
\item maps to a set of unbounded parameters;
\item uses few parameters and is easy to compute; 
\item is bijective, so that each parameter set corresponds uniquely to a covariance matrix and vice versa; and
\item admits an inexpensive inverse mapping.
\end{itemize}

Several such parameterizations have been proposed in the literature, including the Cholesky decomposition \cite{pinheiro_unconstrained_1996}, Givens rotation-based methods \cite{rapisarda_parameterizing_2007,rebonato_most_2011}, and matrix logarithm-based approaches \cite{archakov_new_2021, higham_11_2008, chiu_matrix-logarithmic_1996}.
A comparative study in \cite{bucci_comparing_2022} evaluated these parameterizations 
and found that, on average, matrix logarithm \cite{archakov_new_2021} and Cholesky decomposition \cite{pinheiro_unconstrained_1996} delivered the best performance.

The matrix logarithm parametrization was first used for covariance modeling in \cite{chiu_matrix-logarithmic_1996}, where the entries of the covariance matrix logarithm \cite{higham_11_2008} were regressed on explanatory variables. More recently, \cite{archakov_new_2021} applied this idea to the \textit{correlation} matrix---we call this approach the $\gamma$-transform in our paper. 
We adopt the $\gamma$-transform because it enables flexible prior specification by separating the correlation structure from the standard deviations while satisfying the criteria above.
This separation is particularly advantageous in multi-fidelity settings, where users may have prior knowledge about correlations between model fidelities but little information about the standard deviations, or vice versa.

A covariance matrix $\Sigma$ can be decomposed into a diagonal matrix of standard deviations $D$ and a correlation matrix $R$ via:
\begin{align}\label{eqn:cov_to_corr_sigma}
    \Sigma = D \, R \, D.
\end{align}
Let $\Gamma$ denote the forward $\gamma$-transform and $\mathrm{vecl}(A)$ the vectorization of the lower-triangular elements of a square matrix $A$. The transformation is:
\begin{align}\label{eqn:gamma_mapping}
    \gamma = \Gamma(R) = \mathrm{vecl} \left( \log(R) \right),
\end{align}
where $\log(R)$ is the matrix logarithm. 
This mapping uses $\frac{M(M-1)}{2}$ parameters for the correlation matrix $R$. Combined with the $M$ (log) standard deviations, $\log \sigma = \log(\diag(D))$, this yields a total of $\frac{M(M+1)}{2}$ parameters required to fully specify the covariance matrix $\Sigma$.

For both the forward and inverse $\gamma$-transform, we use numerical implementations from~\cite{archakov_new_2021}, reproduced in Appendix~\ref{app: bij_param}.
Both algorithms are fast, capable of processing hundreds of samples per second. 
The inverse transform is computed using a fixed-point iteration (Algorithm~\ref{alg:gamma-inverse}), for which we adopt a convergence tolerance of $\mathtt{tol}=10^{-7}$.

\subsubsection{Approximate Bayesian inference for \texorpdfstring{$\boldsymbol{\gamma}$}{TEXT}-Gaussian prior}
\label{ss:gamma-inference}

We now perform approximate Bayesian inference for the covariance matrix parameters. While it is possible to infer $\boldsymbol{\gamma}$ and $\log \boldsymbol{\sigma}$ jointly, we adopt explicit simplifying independence assumptions for computational efficiency and flexibility in the Bayesian modeling of $\log \boldsymbol{\sigma}$, leading to the following Bayes' rule expressions:
\begin{align}
    p(\gamma | \hat{\gamma}) = \frac{p(\hat{\gamma}|\gamma) \, p(\gamma)}{p(\hat{\gamma})}, \qquad p(\log \sigma | \log \hat{\sigma}) = \frac{p(\log \hat{\sigma}|\log \sigma) \, p(\log \sigma)}{p(\log \hat{\sigma})},
    \label{e:Bayes}
\end{align}
where $\hat{\gamma}$ and $\log \hat{\sigma}$ are the observed data derived from pilot samples.  
These assumptions include:
\begin{itemize}
\item the priors for $\boldsymbol{\gamma}$ and $\log \boldsymbol{\sigma}$ are independent;
\item the likelihoods for $\boldsymbol{\hat{\gamma}}$ and $\log \boldsymbol{\hat{\sigma}}$ are conditionally independent; and
\item each likelihood depends only on its corresponding parameter set, i.e., likelihood for $\boldsymbol{\hat{\gamma}}$ only depends on ${\gamma}$, and 
likelihood for $\log \boldsymbol{\hat{\sigma}}$ only depends on $\log {\sigma}$.
\end{itemize}

\paragraph{Priors}
We place independent Gaussian priors:
\begin{align}
    p(\gamma) \sim \mathcal{N}(\mu_{0,\gamma},\Sigma_{0,\gamma}), \qquad 
    p(\log \sigma) \sim \mathcal{N}(\mu_{0,\log \sigma},\Sigma_{0,\log \sigma}),
    \label{e:gamma-priors}
\end{align}
where $\mu_{0}$ and $\Sigma_{0}$ denote prior means and covariances.

\paragraph{Likelihood construction}
We assume a Wishart likelihood in covariance-matrix space and map it into the $\gamma$ and $\log \sigma$ spaces, where we estimate Gaussian likelihoods. 
On the data side, we compute the biased sample covariance (needed for Wishart compatibility) from the $N_{\mathrm{pilot}}$ pilot samples:
\begin{align}
\hat{C} =  \frac{1}{N_{\mathrm{pilot}}} \sum_{j=1}^{N_{\mathrm{pilot}}} (f^{(j)} - \bar{f})(f^{(j)} - \bar{f})^\top,
\end{align}
and factor $\hat{C} = \hat{D} \, \hat{R} \, \hat{D}$ to obtain 
$\hat{\gamma} = \Gamma(\hat{R})$ and $
\log \hat{\sigma}=\log ( \mathrm{diag}(\hat{D}))$.

On the model side, given $\gamma$ and $\log \sigma$, we first draw $N_{\mathrm{sim}}$ 
independent \textit{simulated} covariances from the Wishart distribution:
\begin{align}
\hat{C}_{\mathrm{sim}}^{(j)} \Big| \gamma , \log\sigma &\sim \mathrm{Wishart}\left(\frac{1}{N_{\mathrm{pilot}}}D \, \Gamma^{-1}(\gamma)\, D, N_{\mathrm{pilot}}\right),
\label{e:Wishart}
\end{align}
where $D$ is formed from $\log \sigma$.
Each sample is factored into $\hat{C}_{\mathrm{sim}}^{(j)} = \hat{D}_{\mathrm{sim}}^{(j)} \, \hat{R}_{\mathrm{sim}}^{(j)} \, \hat{D}_{\mathrm{sim}}^{(j)}$, then mapped to
\begin{align}
\hat{\gamma}_{\mathrm{sim}}^{(j)} = \Gamma\left(\hat{R}_{\mathrm{sim}}^{(j)}\right), \qquad
\log \hat{\sigma}^{(j)}_{\mathrm{sim}}=\log \left( \mathrm{diag}(\hat{D}_{\mathrm{sim}}^{(j)}) \right).
\end{align}
From these, we compute empirical means and covariances:
\begin{align}
    \mu_{\gamma} = \frac{1}{N_{\mathrm{sim}}} \sum_{j=1}^{N_{\mathrm{sim}}} \hat{\gamma}^{(j)}_{\mathrm{sim}}, 
    \qquad
    \Sigma_{\gamma} = \frac{1}{N_{\mathrm{sim}}-1} \sum_{j=1}^{N_{\mathrm{sim}}} (\hat{\gamma}_{\mathrm{sim}}^{(j)} - \mu_{\gamma})(\hat{\gamma}_{\mathrm{sim}}^{(j)} - \mu_{\gamma})^\top,
    \label{eqn:gamma_likelihood_terms}
\end{align}
and similarly for $\mu_{\log \sigma}$ and $\Sigma_{\log \sigma}$.
Invoking the independence assumptions, we form Gaussian likelihoods:
\begin{align}
p(\hat{\gamma} | \gamma) \sim \mathcal{N}(\mu_{\gamma}(\gamma), \Sigma_{\gamma}(\gamma)), \qquad
p(\log \hat{\sigma} | \log \sigma) \sim \mathcal{N}(\mu_{\log \sigma}(\log \sigma), \Sigma_{\log \sigma}(\log \sigma)),
\label{e:gamma-likelihoods}
\end{align}
where in practice $\Sigma_{\log \sigma}$ is further taken to be diagonal. 

\paragraph{Posterior Update}
With Gaussian priors and likelihoods, the posteriors remain Gaussian in closed form:
\begin{align}\label{eqn:gamma-posterior}
    p(\gamma|\hat{\gamma}) \sim \mathcal{N}(\mu_{\mathrm{post},\gamma},\Sigma_{\mathrm{post},\gamma}), 
    \qquad
    p(\log \sigma|\log \hat{\sigma}) \sim \mathcal{N}(\mu_{\mathrm{post},\log \sigma},\Sigma_{\mathrm{post},\log \sigma}),
\end{align}
where
\begin{align}\label{eqn:gamma-posterior_update}
    \mu_{\mathrm{post},\gamma} = \Sigma_{\mathrm{post},\gamma} \left(  \Sigma_{\gamma}^{-1} \mu_{\gamma} + \Sigma_{0,\gamma\,}^{-1}\mu_{0,\gamma} \right), \qquad
    \Sigma_{\mathrm{post},\gamma} = \left( \Sigma_{\gamma}^{-1}+ \Sigma_{0,\gamma}^{-1}\right)^{-1},
\end{align}
and analogously for $\mu_{\mathrm{post},\log \sigma}$ and $\Sigma_{\mathrm{post},\log \sigma}$.

\paragraph{Remarks}
The independence assumptions are a deliberate modeling choice for computational efficiency and flexibility. A fully joint Gaussian model could capture correlations between $\boldsymbol{\gamma}$ and $\log \boldsymbol{\sigma}$, but would require manipulating larger dense covariance matrices.
Together with efficient $\Gamma(\cdot)$ evaluation and Wishart sampling, the approximate inference is quite fast. We thus recommend a large $N_{\mathrm{sim}}$; for example, we use $N_{\mathrm{sim}}=1000$ in our numerical experiments unless otherwise specified.
Because the mapping $\hat{C} \mapsto \hat{\gamma}$ is nonlinear, the empirical estimates in \eqref{eqn:gamma_likelihood_terms} can be sensitive to outliers. Winsorization or truncation of $\hat{\gamma}$ samples (we use truncation at the 0.2--0.8 quantiles) helps stabilize the posterior. 
The truncation level is a tunable hyperparameter: stronger truncation may produce overconfident posteriors and premature pilot-sampling termination, while weaker truncation may yield broader posteriors and delayed pilot-sampling termination. 
Finally, while we model standard deviations as log-normal here, the independence structure conveniently allows alternative prior-likelihood choices without affecting the treatment of $\boldsymbol{\gamma}$.

\subsection{Other prior choices}\label{ss:other_priors}

There is a significant body of research on Bayesian priors for covariance and correlation matrices; we provide a necessarily non-exhaustive summary here. Extensions of the IW prior have sought to improve its behavior via shrinkage \cite{berger_bayesian_2020} or to increase flexibility by introducing additional scale parameters that separately control the variances and the correlation matrix \cite{omalley_domain_level_2008}. These modifications, however, forfeit conjugacy.

The closest related work to our proposed $\gamma$-Gaussian prior are \cite{leonard_bayesian_1992} and \cite{Hsu_Sinay_Hsu_2010}, which place a multivariate Gaussian prior on the elements of the matrix logarithm of the \textit{covariance} matrix. 
This approach allows a Gaussian likelihood approximation via Bellman's iterative solution to the linear Volterra integral function, retaining conjugacy \cite{leonard_bayesian_1992}. In contrast, \cite{Hsu_Sinay_Hsu_2010} replaces conjugacy with a Laplace approximation or MCMC, with the latter showing better empirical performance. 
A key distinction from our approach is that these priors are placed on the \textit{entire} covariance matrix, making it less straightforward to encode multi-fidelity prior knowledge where one may have separate belief levels about correlations and standard deviations.  

In many probabilistic programming frameworks, the default prior on the correlation matrix is the Lewandowski--Kurowicka--Joe (LKJ) prior \cite{lewandowski_generating_2009}, which is well-behaved and relatively noninformative when its sole parameter is set to $\eta=1$, but is not conjugate.
An informative version was introduced in \cite{martin_informative_2021}, in which each correlation-matrix element penalized by a Gaussian weight: $p(R) \propto \mathrm{LKJ}(R|\eta=1)\times\prod_{j=1}^{\frac{M(M-1)}{2}}\mathcal{N}(R_j;\mu=\mu_j,\sigma=\sigma_j)$, where $R_j$ is an upper triangular element of $R$ with associated prior mean $\mu_j$ and variance $\sigma_j$. 
This formulation allows interpretable, element-wise prior beliefs while preserving PSD, but is neither conjugate nor proper, making both prior and posterior sampling require MCMC.

A notable advantage of the $\gamma$-Gaussian prior is its interpretability: correlations and standard deviations are modeled separately, allowing distinct priors for each. 
For less informative settings, one can map a prior mean correlation matrix to $\boldsymbol{\gamma}$ and assign a diagonal Gaussian prior on $\boldsymbol{\gamma}$ with large variances. For more informative settings, we recommend a procedure similar to that in Section~\ref{ss:gamma-inference}: 
\begin{itemize}
\item specify a prior on the correlation matrix that is easier to interpret (e.g., the modified LKJ prior from \cite{martin_informative_2021}, which allows element-wise mean and variance specification);
\item sample correlation matrices from this prior (via MCMC if needed);
\item transform the samples into parameters $\gamma$ and $\log \sigma$; and 
\item perform moment-matching on the transformed samples to define a Gaussian prior.
\end{itemize}
While sampling from the modified LKJ prior requires MCMC, the cost is negligible compared to high-fidelity model evaluations, making this approach practical for constructing informative $\gamma$-Gaussian priors. 

\section{Adaptive pilot sampling termination}\label{sec:adaptive_pilot_sampling}

We propose a novel adaptive termination procedure for pilot sampling, guided by the Bayesian posterior of the model-output covariance matrix. The procedure is based on a posterior-expected loss that quantifies the inefficiency of a given ACV estimator if pilot sampling were to continue. 

\subsection{Expected estimator loss}\label{sec:expected_loss}
Let $\zeta = \{\alpha, \CA\}$ denote the ACV hyperparameters, and let $\zeta_{\Sigma', B', w'}$ represent the optimal hyperparameters for an ACV estimator when the model-output covariance matrix is $\Sigma'$, the estimator budget $B'$, and model costs are $w'$, as determined via \eqref{eqn:alpha-star-acv} and \eqref{eqn:mxmcpy}. 

We define the \textit{best-case} estimator as the (infeasible) estimator that has (1) access to the oracle (exact) covariance matrix $\Sigma_{\mathrm{or}}$ and costs vector $w_{\mathrm{or}}$, 
and (2) the entire budget $B_{\mathrm{tot}}$ allocated to ACV estimation (i.e., no budget spent on pilot sampling). We denote this best-case estimator by $\boldsymbol{\Tilde{Q}}(\boldsymbol{z}; \zeta_{\Sigma_{\mathrm{or}}, B_{\mathrm{tot}}, w_{\mathrm{or}}})$.
In most practical settings, the model costs are known or easily estimated \textit{a priori}, so for simplicity we omit explicit dependence 
on $w$ for the remainder of this section.

\paragraph{Loss definition}
The key idea is that an estimator whose hyperparameters are optimized using nominal values of $(\Sigma', B')$ will generally have a different \textit{actual variance} when evaluated under the true covariance matrix $\Sigma_{\mathrm{or}}$.
Specifically, if $\Sigma' \neq \Sigma_{\mathrm{or}}$, there is a covariance \textit{accuracy loss}, and if $B' \neq B_{\mathrm{tot}}$, there is a budget  \textit{cost loss}---we will analyze these contributions further in Section~\ref{sec:loss_decomposition}.

The true variance of $\boldsymbol{\Tilde{Q}}(\boldsymbol{z}; \zeta_{\Sigma', B'})$ under $\Sigma_{\mathrm{or}}$ can be written as:
\begin{align}\label{eqn:true_estimator_variance}
     \Var\left[ \boldsymbol{\Tilde{Q}}(\boldsymbol{z}; \zeta_{\Sigma', B'})\right] \Big| \Sigma_{\mathrm{or}}  &=  \left(\Var[\boldsymbol{\hat{Q}}_{0}] \Big| \Sigma_{\mathrm{or}}\right) - \alpha_{\Sigma', B'}^{\top}\Big(\Cov[\boldsymbol{\Delta}, \boldsymbol{\Delta}] \Big| \Sigma_{\mathrm{or}}\Big)^{-1}\alpha_{\Sigma', B'} \\ 
     &\hspace{1.5em}+2 \alpha_{\Sigma', B'}^{\top} \left(\Cov[\boldsymbol{\Delta}, \boldsymbol{\hat{Q}}_{0}] \Big| \Sigma_{\mathrm{or}}\right) \nonumber,
\end{align}
where the conditioning notation in $\left(\Var[\boldsymbol{\hat{Q}}_{0}] \Big| \Sigma_{\mathrm{or}}\right)$ denotes that the variance is evaluated with $\Sigma_{\mathrm{or}}$ (and similarly for other variance and covariance terms), and $\alpha_{\Sigma', B'}$ are the optimal weights under $(\Sigma',B')$. 

We define the loss (or inefficiency) of 
using $(\Sigma',B')$
instead of $(\Sigma_{\mathrm{or}},B_{\mathrm{tot}})$ as:
\begin{align}\label{eqn:estimator_loss}
    \CL\left(\Sigma', B', \Sigma_{\mathrm{or}}, B_{\mathrm{tot}}\right) = \underbrace{\left(\Var[\boldsymbol{\Tilde{Q}}(\boldsymbol{z}; \zeta_{\Sigma', B'})] \Big| \Sigma_{\mathrm{or}} \right)}_{\text{actual estimator variance}} - \underbrace{\left( \Var[\boldsymbol{\Tilde{Q}}(\boldsymbol{z}; \zeta_{\Sigma_{\mathrm{or}}, B_{\mathrm{tot}}})] \Big| \Sigma_{\mathrm{or}} \right)}_{\text{best-case estimator variance}}.
\end{align}
This quantity measures the increase in ACV variance due to pilot sampling, incorporating both budget and covariance estimation errors. 

\paragraph{Posterior-expected loss}
In practice, $\Sigma_{\mathrm{or}}$ is unknown. Instead, we compute the \textit{expected loss} by taking the expectation of \eqref{eqn:estimator_loss} over the random variable $\boldsymbol{\Sigma}$ that represents the possible oracle covariance matrices, given pilot sample data $D$:
\begin{align}
    \ExpLoss\big(\Sigma', B', p_{\boldsymbol{\Sigma}| {D}}, B_{\mathrm{tot}}\big) &= \EE_{\boldsymbol{\Sigma}| {D}} \left[ \CL(\Sigma', B', \boldsymbol{\Sigma}, B_{\mathrm{tot}}) \right] \label{eqn:expected_loss}  \\
    &= \EE_{\boldsymbol{\Sigma}| D} \left[ \left(\Var[\boldsymbol{\boldsymbol{\Tilde{Q}}}(\boldsymbol{z}; \zeta_{\Sigma', B'})] \Big| \boldsymbol{\Sigma} \right) - \left( \Var[\boldsymbol{\boldsymbol{\Tilde{Q}}}(\boldsymbol{z}; \zeta_{\Sigma, B_{\mathrm{tot}}})] \Big| \boldsymbol{\Sigma} \right) \right]. 
\end{align}
Here, the variance terms are now averaged over the posterior uncertainty in $\boldsymbol{\Sigma}$. 

\paragraph{Monte Carlo approximation}
We approximate \eqref{eqn:expected_loss} using Monte Carlo:
\begin{align}
    \ExpLoss\big(\Sigma', B', p_{\boldsymbol{\Sigma}| D}, B_{\mathrm{tot}}, \big) \approx \frac{1}{N_{\mathrm{MC}}} \sum_{i=1}^{N_{\mathrm{MC}}} \CL(\Sigma', B', \Sigma^{(i)}, B_{\mathrm{tot}}),
    \label{eqn:MC_loss}
\end{align}
where $\Sigma^{(i)}\sim p(\Sigma | D)$ are independent posterior draws. 

Evaluating \eqref{eqn:MC_loss} requires solving the ACV optimization problem in \eqref{eqn:mxmcpy}--\eqref{eqn:mxmcpy_constraint} $N_{\mathrm{MC}}$ times to obtain $\zeta_{\Sigma,B_{\mathrm{tot}}}$ for each $\Sigma^{(i)}$ sample, plus an additional solve for $\zeta_{\Sigma',B'}$. 
The MXMCPy tool \cite{bomarito_multi_2020} can perform these computations quickly relative to a high-fidelity model run, and the procedure is embarrassingly parallelizable. 

To further reduce computational cost, we fix the ACV estimator's allocation strategy to whichever family (e.g., `ACVIS', `WRDIFF', `MFMC'; see \cite{bomarito_multi_2020} for a full list) performs best for $(\Sigma',B')$. This restriction makes the expected loss equivalent to the inefficiency due to pilot sampling within that strategy family. In our experiments, this simplification did not meaningfully change termination results but yielded significant speedups.

\subsection{Interpreting the expected loss}\label{sec:loss_decomposition}

A key benefit of the expected loss is that it admits a natural decomposition into two components: (1) accuracy loss: inefficiency due to inaccuracy in the covariance estimate; and (2) cost loss: inefficiency due to the budget consumed by pilot sampling. 

\paragraph{Loss decomposition}

We can rewrite the total estimator loss from \eqref{eqn:estimator_loss} as:

\bigskip
\bigskip
\begin{align} \label{e:decomposition}
  \CL(\Sigma', B', \Sigma_{\mathrm{or}}, B_{\mathrm{tot}}) &= \color{Fuchsia}{\eqnmarkbox[colorNb]{boxlabel1}{\left(\Var[\boldsymbol{\Tilde{Q}}(\boldsymbol{z}; \zeta_{\Sigma', B'})] \Big| \Sigma_{\mathrm{or}} \right) - \left( \Var[\boldsymbol{\Tilde{Q}}(\boldsymbol{z}; \zeta_{\Sigma_{\mathrm{or}}, B'})] \Big| \Sigma_{\mathrm{or}} \right)}} \\&+ \color{DarkGoldenrod3}{\eqnmarkbox[colorY]{boxlabel2}{\left( \Var[\boldsymbol{\Tilde{Q}}(\boldsymbol{z}; \zeta_{\Sigma_{\mathrm{or}}, B'})] \Big| \Sigma_{\mathrm{or}} \right) - \left( \Var[\boldsymbol{\Tilde{Q}}(\boldsymbol{z}; \zeta_{\Sigma_{\mathrm{or}}, B_{\mathrm{tot}}})] \Big| \Sigma_{\mathrm{or}} \right).}} \nonumber
\end{align}
\vspace{15mm}
\annotate[yshift=1em]{above}{boxlabel1}{\large\textbf{Accuracy Loss}}
\annotate[yshift=-1em]{below}{boxlabel2}{\large\textbf{Cost Loss}}

The second and third terms in \eqref{e:decomposition} cancel, leaving the first and last terms unchanged from \eqref{eqn:estimator_loss}. 

\paragraph{Accuracy loss}
The first two terms in \eqref{e:decomposition} together define the accuracy loss:
\begin{align}
    \CL_{A}(\Sigma', B', \Sigma_{\mathrm{or}}) = \left(\Var[\boldsymbol{\Tilde{Q}}(\boldsymbol{z}; \zeta_{\Sigma', B'})] \Big| \Sigma_{\mathrm{or}} \right) - \left( \Var[\boldsymbol{\Tilde{Q}}(\boldsymbol{z}; \zeta_{\Sigma_{\mathrm{or}}, B'})] \Big| \Sigma_{\mathrm{or}} \right),
\end{align}
which compares two estimators with the same budget $B'$ but different covariance matrices: $\Sigma'$ versus $\Sigma_{\mathrm{or}}$. This isolates the effect of covariance estimation error on ACV variance reduction.

\begin{lemma}\label{lem:accuracy_loss_positive}
    If
    $\Sigma' \neq \Sigma_{\mathrm{or}}$, then $\CL_{A}(\Sigma', B', \Sigma_{\mathrm{or}}) \geq 0$.
\end{lemma}
\begin{proof}
    Following \eqref{eqn:alpha-star-acv} and \eqref{eqn:mxmcpy}--\eqref{eqn:mxmcpy_constraint}, the variance-minimizing hyperparameters under $\Sigma_{\mathrm{or}}$ and $B'$ are $\zeta_{\Sigma_{\mathrm{or}},B'}=\{\alpha^\ast,\CA^\ast\}$. For any $\zeta'\neq\zeta_{\Sigma_{\mathrm{or}},B'}$, 
    the resulting variance therefore must be no smaller, hence $\CL_{A}(\Sigma', B', \Sigma_{\mathrm{or}}) \geq 0$.
\end{proof}

\paragraph{Cost loss}
The last two terms in \eqref{e:decomposition} together define the cost loss:
\begin{align}\label{eqn:cost_loss}
    \CL_{C}(B', \Sigma_{\mathrm{or}}, B_{\mathrm{tot}}) = \left( \Var[\boldsymbol{\Tilde{Q}}(\boldsymbol{z}; \zeta_{\Sigma_{\mathrm{or}}, B'})] \Big| \Sigma_{\mathrm{or}} \right) - \left( \Var[\boldsymbol{\Tilde{Q}}(\boldsymbol{z}; \zeta_{\Sigma_{\mathrm{or}}, B_{\mathrm{tot}}})] \Big| \Sigma_{\mathrm{or}} \right),
\end{align}
which compares two estimators with the same covariance matrix $\Sigma_{\mathrm{or}}$ but different budgets: $B'$ versus $B_{\mathrm{tot}}$.
This isolates the effect of budget reduction from pilot sampling. 

\begin{lemma}\label{lem:cost_loss_positive}
    If $B'<B_{\mathrm{tot}}$, 
    then $\CL_{C}(B', \Sigma_{\mathrm{or}}, B_{\mathrm{tot}}) \geq 0$.
\end{lemma}
\begin{proof}
    With $\Sigma_{\mathrm{or}}$ fixed, both estimators use the same 
    $\alpha$ from \eqref{eqn:alpha-star-acv}, differing only in $\CA$ from \eqref{eqn:mxmcpy}--\eqref{eqn:mxmcpy_constraint}. The feasible domain for $\CA$ under $B_{\mathrm{tot}}$ contains that under $B'$, so the optimal variance for $B_{\mathrm{tot}}$ is no greater than that for $B'$, and hence $\CL_{C}(B', \Sigma_{\mathrm{or}}, B_{\mathrm{tot}}) \geq 0$.
\end{proof}

\paragraph{Posterior-expected components}
Exact evaluation of $\CL_{A}$ and $\CL_{C}$ requires $\Sigma_{\mathrm{or}}$, which is unknown. We therefore take the expectations over the posterior $p(\Sigma|D)$:
\begin{align}\label{eqn:expected_accuracy_loss}
    \overline{\CL}_{A}(\Sigma', B', p_{\boldsymbol{\Sigma} | D}) &= \EE_{\boldsymbol{\Sigma} | D} \left[ \CL_{A}(\Sigma', B', \boldsymbol{\Sigma}) \right] \\&= \EE_{\boldsymbol{\Sigma} | D} \left[ \left(\Var[\boldsymbol{\Tilde{Q}}(\boldsymbol{z}; \zeta_{\Sigma', B'})] \Big| \boldsymbol{\Sigma} \right) - \left( \Var[\boldsymbol{\Tilde{Q}}(\boldsymbol{z}; \zeta_{\Sigma, B'})] \Big| \boldsymbol{\Sigma} \right) \right], \nonumber\\
\label{eqn:expected_cost_loss}
    \overline{\CL}_{C}(B', p_{\boldsymbol{\Sigma} | D}, B_{\mathrm{tot}}) &= \EE_{\boldsymbol{\Sigma} | D} \left[ \CL_{C}(B', \boldsymbol{\Sigma}, B_{\mathrm{tot}}) \right] \\&= \EE_{\boldsymbol{\Sigma} | D} \left[ \left( \Var[\boldsymbol{\Tilde{Q}}(\boldsymbol{z}; \zeta_{\Sigma, B'})] \Big| \boldsymbol{\Sigma} \right) - \left( \Var[\boldsymbol{\Tilde{Q}}(\boldsymbol{z}; \zeta_{\Sigma, B_{\mathrm{tot}}})] \Big| \boldsymbol{\Sigma} \right) \right]. \nonumber
\end{align}
Each term can be approximated via MC in a similar manner as \eqref{eqn:mc}. 
$\overline{\CL}_{A}$ can be interpreted as an \textit{exploitation loss}: it generally decreases with more pilot samples as $\Sigma'$ approaches $\Sigma_{\mathrm{or}}$. $\overline{\CL}_{A}$ can be interpreted as an \textit{exploration loss}: it generally increases with more pilot samples because less budget remains for the final ACV estimation. This decomposition quantifies the trade-off between exploration (pilot sampling) and exploitation (final ACV estimation), providing a principled basis for adaptive stopping.

\subsection{Projecting covariance posteriors}\label{sec:projected-posteriors}

To approximate the expected loss for future iterations of pilot sampling, we need to project the posterior distribution that would result after collecting additional, as-yet-unobserved data. 

Let $D$ denote the current set of pilot samples and $p(\Sigma | D)$ the corresponding posterior over the covariance matrix. We write $p(\Sigma | D_{n})$ for the \textit{projected posterior} after $n$ additional pilot-sampling iterations. Since $D_{n} \notin D$ is unknown, this projection is necessarily unknown. 
When the prior-likelihood pair is conjugate, such as with the IW or the $\gamma$-Gaussian priors introduced in Section~\ref{sec:bayes_inf}, these projections are straightforward and computationally inexpensive. 
In practice, when projecting only a few pilot sampling iterations ahead, the resulting approximate reduction in posterior uncertainty is sufficiently accurate to guide adaptive stopping decisions.

\subsubsection{Inverse Wishart projected posteriors}
We consider an IW posterior parameterized by $H=H_0+S\times N_{\mathrm{pilot}}$, $\nu=\nu_0+N_{\mathrm{pilot}}$, where $H$ is the updated scatter matrix (posterior ``center''), $\nu$ encodes the 
total number of pseudo-pilot samples from both the prior and observed data, and $S$ is the sample scatter matrix from the observed pilot data. 

If we assume that $n$ additional pilot samples will have the same scatter matrix $S$ as the existing data, the conjugate update formulas for the projected posterior become:
\begin{align}
    \nu_n &= \nu_0 + N_{\mathrm{pilot}} + n,\\
    H_n &= H_0 + S\times(N_{\mathrm{pilot}} + n),
\end{align}
yielding an analytically updated IW posterior $\mathrm{IW}(\nu_n,H_n)$ without requiring actual new samples.

\subsubsection{\texorpdfstring{$\boldsymbol{\gamma}$}{TEXT}-Gaussian projected posteriors}
For the $\gamma$-Gaussian prior, we project posteriors using the moment-matching procedure from Section~\ref{ss:gamma-inference} but artificially tighten the likelihood parameters in \eqref{eqn:gamma_likelihood_terms} to reflect the assumed $n$ extra samples. We assume the unseen data have the same sample statistics as the observed data.
Concretely, we update the Wishart distribution for generating simulated covariance matrices in \eqref{e:Wishart} to $\mathrm{Wishart}\left(\frac{1}{N_{\mathrm{pilot}}+n}D\Gamma^{-1}(\gamma)D,N_{\mathrm{pilot}}+n\right)$.

\subsection{Hyperparameters for final ACV estimation} \label{sec:cov_point_estimate}

In the next section, we will address pilot-sampling termination. Once pilot sampling is complete, the remaining budget is allocated entirely to the final ACV estimation, which requires selecting the final ACV hyperparameters. 
Given the posterior distribution $p(\Sigma|D)$ over the covariance, one could, in principle, draw covariance samples from the posterior, optimize $\zeta$ for each sample, and average the resulting ACV estimates, yielding a posterior-expected ACV estimate. However, this would require multiple ACV evaluations---each with distinct allocations---and is computationally prohibitive. Instead, we perform a single final ACV estimation using a point estimate of the covariance. 

We recommend using the posterior mean (or mode) of $\boldsymbol{\Sigma}$, which incorporates prior information and thus regularizes the estimate. This is particularly valuable in low-sample regimes or when informative priors are available. This approach contrasts with the conventional estimate based on the unbiased sample covariance in \eqref{eqn:cov}, which can be highly variable when pilot data are scarce.

In our numerical examples, we adopt the posterior mean.
For the IW posterior, the mean is given in \eqref{eqn:iw_mean}. For the $\gamma$-Gaussian posterior, the mean can be approximated either by MC sampling
or analytically as
$\overline{\Sigma}_{\gamma,\sigma}=
D(\overline{\sigma})\Gamma(\overline{\gamma})D(\overline{\sigma})$, where $\overline{\gamma}$ is the posterior means of $\boldsymbol{\gamma}$, and $\overline{\sigma}$ is obtained by exponentiating the posterior mean of $\log \boldsymbol{\sigma}$.

\subsection{Adaptive pilot-sampling termination}\label{sec:adaptive_algorithm}

We extend the expected loss from \eqref{eqn:expected_loss} to develop an adaptive algorithm that projects the expected loss over future pilot-sampling collections and determines whether to stop pilot sampling based on these projections. 

Let $D$ be the current set of pilot samples, with remaining budget $B'$ (from the total budget $B_{\mathrm{tot}}$), current posterior distribution $p(\Sigma|D)$, and a point estimate $\Sigma'$ for ACV hyperparameters (e.g., the posterior mean; see Section~\ref{sec:cov_point_estimate}). The \textit{projected expected loss} after $n$ additional pilot samples is:
\begin{align}\label{eqn:projected_expected_loss}
    \ExpLoss_{n}\left(\Sigma', B'_{n}, p_{\boldsymbol{\Sigma}| D}, B_{\mathrm{tot}}\right) = \ExpLoss\left(\Sigma', B', p_{\boldsymbol{\Sigma}| D_n}, B_{\mathrm{tot}}\right),
\end{align}
where $B'_{n}$ is the remaining budget after $n$ more pilot samples, and $p_{\boldsymbol{\Sigma}| D_{n}}$ is the projected posterior (see Section~\ref{sec:projected-posteriors}). 

The pilot sampling termination problem based on the current data $D$ involves minimizing the projected expected loss:
\begin{align}\label{eqn:n_star_new}
    \min_{\substack{ B'_{n}\leq B_{\mathrm{tot}} \\ n \, \in \, \mathbb{Z}^{+}}} \,\,\ExpLoss_{n}\left(\Sigma', B'_n, p_{\boldsymbol{\Sigma}| D_n}, B_{\mathrm{tot}}\right),
\end{align}
which balances improved uncertainty in $\Sigma'$ against the reduced budget $B'_{n}$ available for final ACV estimation.
Because posterior projections rely on unseen data, their accuracy decreases as $n$ grows. 
This motivates an \textit{adaptive algorithm} (Algorithm~\ref{alg:pilot_sampling_termination}) 
that re-estimates a solution to \eqref{eqn:n_star_new} as new pilot samples are collected.

\begin{algorithm}[htbp]
\caption{Adaptive pilot-sampling termination}\label{alg:pilot_sampling_termination}
\begin{algorithmic}[1]
\Require Model ensemble $f(Z)$, costs vector $w$, total budget $B_{\mathrm{tot}}$, prior distribution over covariance matrix $p_{\boldsymbol{\Sigma}}$, pilot-sampling batch size $k$, projection steps $n_{\mathrm{steps}}$
\vspace{2mm}
\Ensure Estimate of high-fidelity mean ${\Tilde{Q}}$, samples of estimator errors based on posterior covariance $\{\Var[\boldsymbol{\Tilde{Q}}]\}_{i=1}^{N}$, total number of pilot samples drawn $N^{\ast}_{\mathrm{pilot}}$ 
\vspace{2mm}
\State Initialize $n_{\mathrm{running}} = 0$, $D=\emptyset$, and $B'=B_{\mathrm{tot}}$
\While{$B' > 0$}
    \If{$n_{\mathrm{running}}=0$}
        \State Draw $M+1$ independent samples of $\boldsymbol{Z}$, evaluate models $f(Z)$,
        assign to $D$ 
        \State $n_{\mathrm{running}} \gets M + 1$
    \Else
        \State Draw $k$ independent samples of $\boldsymbol{Z}$ and concatenate $f(Z)$ onto $D$
        \State $n_{\mathrm{running}} \gets n_{\mathrm{running}}+k$
    \EndIf
    \State $B' \gets B_{\mathrm{tot}} - \sum_{m=0}^{M} w_{m} n_{\mathrm{running}}$
    \State $p_{\boldsymbol{\Sigma} | D} \gets$ \texttt{BayesInf}($p_{\boldsymbol{\Sigma}},p_{{D} | \boldsymbol{\Sigma}}, D$) following Section~\ref{sec:bayes_inf}
    \State Select $\Sigma'$ according to posterior samples (e.g., posterior mean), per Section~\ref{sec:cov_point_estimate}
    \State Compute expected loss $L = \overline{\CL}(\Sigma', B', p_{\boldsymbol{\Sigma} | D}, B_{\mathrm{tot}})$ at current step using \eqref{eqn:expected_loss}
    \State Initialize $\underline{L}=\emptyset$
    \For{$i = 1,\ldots,n_{\mathrm{steps}}$}
        \State $K \gets k\times i$
         \State Draw samples from the $K$-sample projected posterior $p(\Sigma | D_{K})$ following Section~\ref{sec:projected-posteriors}
        \State $B_{K}' = B_{\mathrm{tot}} - \sum_{m=0}^{M} w_{m} (n_{\mathrm{running}}+K)$
        \State Compute projected expected loss $L_{K} = \overline{\CL}_{K}(\Sigma', B_{K}', p_{\boldsymbol{\Sigma} | D_K}, B_{\mathrm{tot}})$ in \eqref{eqn:projected_expected_loss}
        \State Concatenate $L_K$ onto $\underline{L}$
    \EndFor
    \If{$L<\mathtt{all}(\underline{L})$}
       \State \textbf{break}
    \EndIf
\EndWhile
\State $N_{\mathrm{pilot}}^{\ast}=n_{\mathrm{running}}$
\State ${\Tilde{Q}} \gets \texttt{ACV}(D, \Sigma', B', w)$ following Section~\ref{sec:background}
\State Draw $N$ MC samples from posterior $p(\Sigma | D)$ and compute $\{\Var[\boldsymbol{\Tilde{Q}}]\}_{i=1}^{N}$ via \eqref{eqn:true_estimator_variance}
\vspace{2mm}
\end{algorithmic}
\end{algorithm}

Our implementation to approximately solve \eqref{eqn:n_star_new} is as follows. Consider the case where each pilot-sampling iteration draws in batches of $k$ samples. The algorithm projects out to a horizon of $n_{\mathrm{steps}}$ such batches.
At each update, projected losses over the 
next $n_{\mathrm{steps}}$ iterations are compared to the current loss. If all projected losses exceed the current loss, pilot sampling terminates.
We recommend the following when setting these algorithm hyperparameters:
\begin{itemize}
\item $n_{\mathrm{steps}}$: Keep small when projected loss curves are smooth (enabling faster computation). Increase when curves are noisy to avoid premature termination.

\item $k$: Use small $k$ (even $k=1$) when budgets are small. Increase when budgets are larger to reduce termination checks.
\end{itemize}

Additionally, we suggest a number of implementation techniques that can improve the algorithm performance.
\begin{itemize}
\item Fix random seed: When sampling from $p(\Sigma| D)$ and $p(\Sigma| D_n)$, fixing the seed improves comparability by eliminating MC variability across loss differences. 

\item Adaptive $N_{\mathrm{MC}}$: 
To better distinguish between the values of $\ExpLoss$ and $\ExpLoss_n$, in the first iteration, compute the standard deviation of the MC estimate of $\ExpLoss$. Adjust $N_{\mathrm{MC}}$ so the estimator's standard deviation is at least one order of magnitude smaller than their difference.

\item Stabilizing $\gamma$-Gaussian projections: For $\gamma$-Gaussian priors, fix the log-standard-deviation means between $p(\Sigma | D)$ and each projected $p(\Sigma | D_n)$. This prevents large mean shifts that could mask the changes in posterior uncertainty.
\end{itemize}
\vspace{10mm}
A further advantage over non-Bayesian approaches is that our method yields posterior predictive samples of the estimator variance, $(\Var[\boldsymbol{\Tilde{Q}}] |\Sigma)$, rather than a point estimate. This enables risk-aware decision making using the full predictive uncertainty and avoids the pitfalls illustrated in the right panel of Figure~\ref{fig:pilot_study}.

\section{Numerical experiments}\label{sec:demos}
We now evaluate the performance of our proposed adaptive pilot-sampling algorithm. 
Our experiments progresses in two stages: first, a simple monomial benchmark problem commonly used in multi-fidelity UQ literature; and second, a Darcy flow PDE where the ensemble is composed of neural network surrogate models. 
These two test cases allow us to assess both controlled synthetic settings and more realistic multi-fidelity scenarios. 

Across all experiments, we compare our method against several baselines and oracle benchmarks:
\\
\begin{center}
\begin{tabular}{@{} l p{0.7\linewidth} @{}}
\textbf{Acronym} & \textbf{Description} \\
\hline

MC & Single-fidelity MC using only the highest-fidelity model. \vspace{2mm} \\

MLMC-BEST & Best-case MLMC estimator from \cite{giles_multilevel_2015}, with oracle covariance and the full budget dedicated to MLMC evaluation. This represents an oracle baseline optimal in budget and covariance knowledge, but limitated by how $\alpha$ and $\CA$ are chosen. \vspace{2mm} \\

ACV-BEST & Best-case ACV estimator, with oracle covariance and the full budget dedicated to ACV evaluation. \vspace{2mm} \\

ADAPT-IW & Proposed adaptive pilot-sampling method (Algorithm~\ref{alg:pilot_sampling_termination}) using the IW prior and associated inference. \vspace{2mm} \\

ADAPT-GG & Proposed adaptive pilot-sampling method (Algorithm~\ref{alg:pilot_sampling_termination}) using the $\gamma$-Gaussian prior and associated inference but with diagonal prior covariance, i.e., each $\gamma$ parameter is independent. 
\vspace{2mm} \\

ADAPT-GGMVN & Proposed adaptive pilot-sampling method (Algorithm~\ref{alg:pilot_sampling_termination}) using the $\gamma$-Gaussian prior and associated inference but with full prior covariance. 
\end{tabular}
\end{center}

\subsection{Case 1: Monomial benchmark}\label{sec:monomial}

\subsubsection{Problem setup}
\label{sss:c1_setup}

We begin with a simple multi-fidelity benchmark problem in which models are defined by monomials of decreasing order: 
\begin{align}
f_m(\boldsymbol{z})= \boldsymbol{z}^{5-m}, \qquad m=0,\ldots,M-1, \qquad \boldsymbol{z} \sim \mathcal{U}(0,1).
\end{align}
We take $M=4$, with a high-fidelity model ($f_0$, polynomial degree $5$) and 
three lower-fidelity approximations ($f_1,f_2,f_3$). The cost vector is $w = \{10^{-m}\}_{m=0}^{M-1}$. 
This benchmark is widely used in the multi-fidelity UQ literature, and here we use it to investigate how total budget, prior family, and prior informativeness affect the proposed adaptive pilot-sampling termination.

As a baseline, we specify an uninformative prior.
For IW (ADAPT-IW), we assume the user believes the models are highly correlated, with a prior correlation:
\begin{align}
    R_0 &= 
    \begin{bmatrix}
    1     & 0.975 & 0.95  & 0.925 \\
    0.975 & 1     & 0.95  & 0.95  \\
    0.95  & 0.95  & 1     & 0.95  \\
    0.925 & 0.95  & 0.95  & 1     
    \end{bmatrix},\label{eqn:prior_mean_corr}
\end{align}
and prior standard deviations
\begin{align}
    {\sigma}_0=\left[ 0.1, 0.1, 0.1, 0.1 \right]^{\top},
\end{align}
yielding a prior covariance 
$\Sigma_0 = \mathrm{diag}({\sigma}_0) \, R_0 \, \mathrm{diag}({\sigma}_0)$.
To represent relatively weak prior certainty, we set $\nu_0=6$. Using \eqref{eqn:iw_mean}, the scale parameter is then $H_0  =  \Sigma_0 (\nu_0 - M - 1)$, giving the overall prior distribution $\mathrm{IW}(H_0,\nu_0)$.

For the $\gamma$-Gaussian cases (ADAPT-GG, ADAPT-GGMVN),
we consider two versions: 
\begin{itemize}
\item ADAPT-GG: diagonal prior covariance (i.e., independent normals on each $\boldsymbol{\gamma}$); and
\item ADAPT-GGMVN: full prior covariance for $\boldsymbol{\gamma}$.
\end{itemize}
For both, we simply map $R_0$ from \eqref{eqn:prior_mean_corr} to the $\gamma$ space to establish the prior mean: 
\begin{align}
\mu_{\gamma,0} = [1.5883, 1.14386  , 0.6994, 0.9291, 1.1858, 1.2053]^{\top}.
\end{align}
For ADAPT-GGMVN, the prior for $\boldsymbol{\gamma}$ becomes $\boldsymbol{\gamma}\sim\mathcal{N}(\mu_{\gamma,0},\mathrm{diag}([1,1,1,1])$.
For ADAPT-GG, each prior component is $\boldsymbol{\gamma}_i\sim\mathcal{N}(\mu_{\gamma,0,i},1)$ for $i=0,\ldots,5$. The standard deviation prior in both cases is $\log(\boldsymbol{\sigma}_{m})\sim \mathcal{N}(0.1,1)$ for $m=0,\ldots,3$.

\subsubsection{Results}
We first consider a single run of the ADAPT-GGMVN method with a total budget equivalent to $200$ pilot samples, i.e., $B_\mathrm{tot}=200 \sum_m w_m$. 
With 6 initial pilot samples, we project the expected loss for various horizons, decomposing it into its accuracy and cost components. 
Figure~\ref{fig:monomial_loss_decomp} shows that 
for very small sample sizes, accuracy loss dominates the total loss. As the total sample size increases, cost loss quickly becomes dominant. 
Interestingly, once the total budget is nearly exhausted, accuracy loss begins to rise again, reflecting that even small discrepancies between $\Sigma'$ and $\Sigma_\mathrm{or}$ can yield larger differences in estimator variance when the ACV budget remaining after pilot sampling is limited. Nevertheless, on the log-scale axis, accuracy loss remains a minor contributor overall.

\begin{figure}[htbp]
    \centering
    \includegraphics[width=0.75\linewidth]{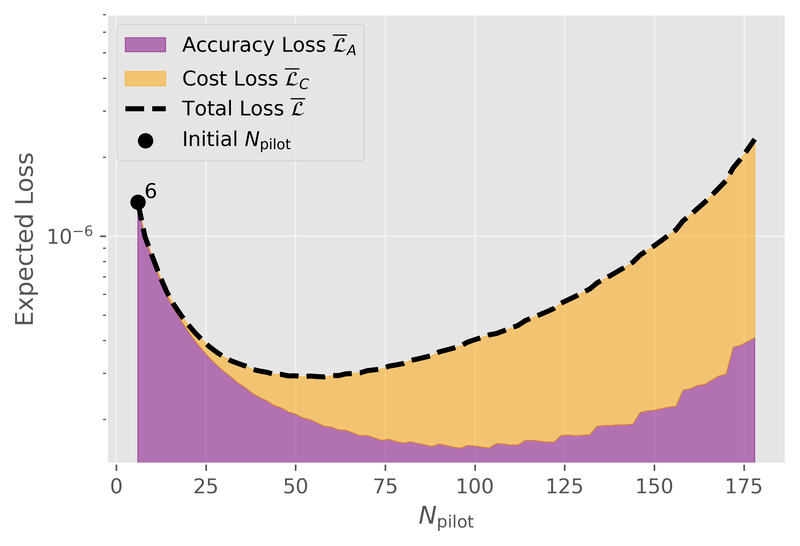}
    \caption{Decomposition of the projected expected loss into \textbf{\textcolor{Fuchsia}{accuracy loss}} and \textbf{\textcolor{DarkGoldenrod3}{cost loss}} components for the ADAPT-GGMVN method. Projections are shown from an initial $N_\mathrm{pilot}=6$ across varying sample sizes, using a single random seed and total budget equivalent to $200$ pilot samples. }
    \label{fig:monomial_loss_decomp}
\end{figure}

As additional pilot samples are drawn in batches, the projected losses are re-estimated and the termination criterion in Algorithm~\ref{alg:pilot_sampling_termination} is re-evaluated. Figure~\ref{fig:monomial_curves} illustrates these projections, with dashed lines indicating the projected loss from different starting sample sizes. While the projections do not exactly match the true loss curve, they are reasonably accurate over short horizons. Since the adaptive algorithm only checks whether the expected loss is projected to increase after the next $n_{\mathrm{steps}}$ batches, the method remains reliable provided $k$
is 
not too large. The histograms in Figure~\ref{fig:monomial_curves} further illustrate how posterior uncertainty propagates into projected estimator variances, with the distributions concentrating toward the true variance as more pilot samples are collected.

\begin{figure}[htbp]
    \centering
    \includegraphics[width=0.99\linewidth]{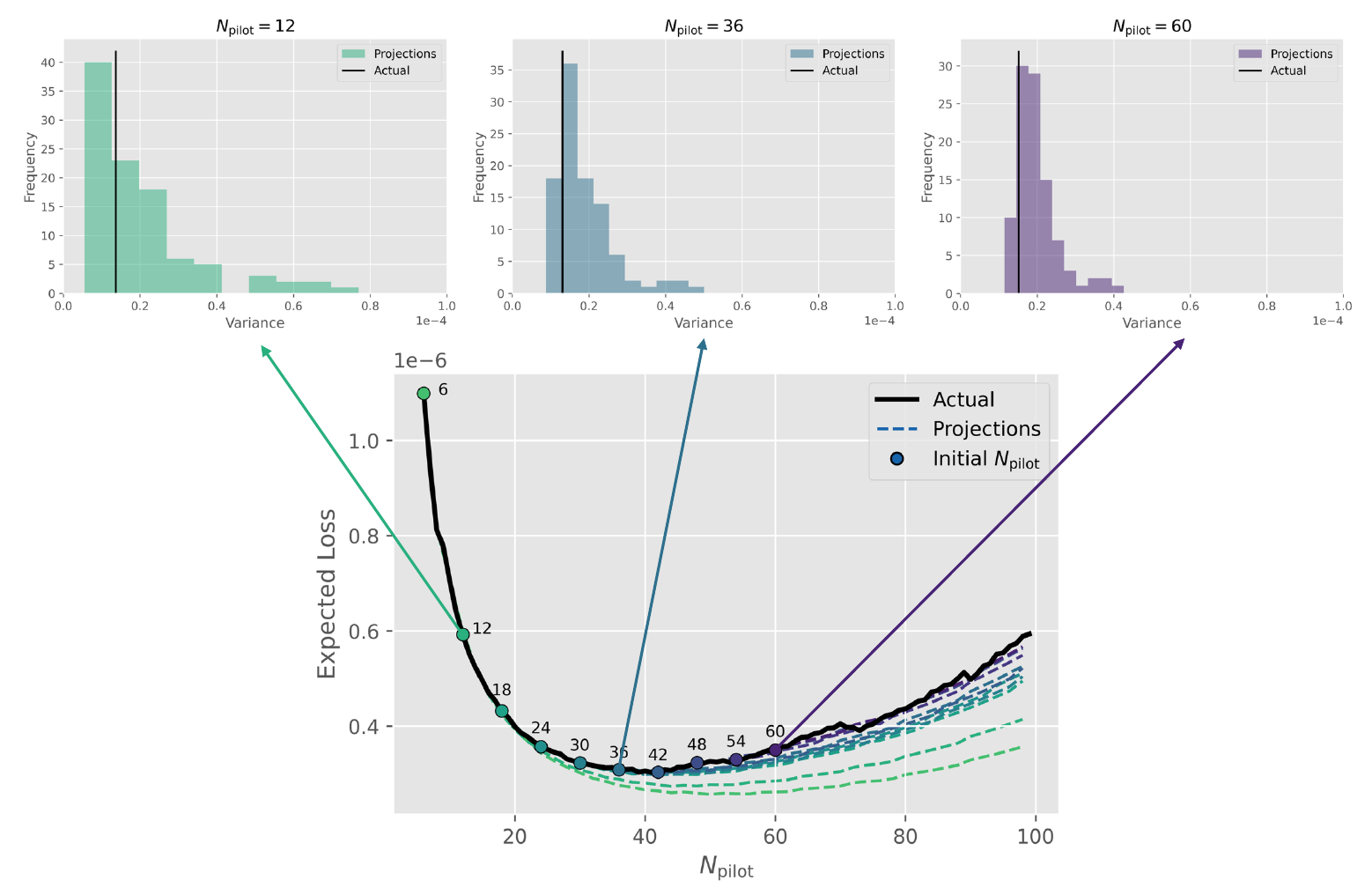}
    \caption{Projected and actual expected loss curves for the ADAPT-GGMVN method under a single random seed with a total budget equivalent to $200$ pilot samples. 
    Histograms of estimated variances are compared against the true variance at $N_\mathrm{pilot}=12$, 36, and 60. As additional samples are collected, the histograms concentrate toward the true value.
    }
    \label{fig:monomial_curves}
\end{figure}

To evaluate the overall algorithm performance and study the effect of the total budget, we repeat the algorithm for each prior across 20 trials with budgets equivalent to 50, 100, and 200 pilot samples. 
Each expected-loss evaluation uses 500 MC samples, the batch size is $k=2$, and we select $n_{\mathrm{steps}}=1$ due to the smoothness in the projected loss curves.
These are extremely limited budgets compared to typical multi-fidelity estimation studies, where substantially pilot sampling is usually recommended. For comparison, we also compute the ACV-BEST and MLMC-BEST estimator variances, both using oracle covariance and dedicating the full budget to final evaluation. 
MLMC-BEST is particularly useful as it shows the variance achievable by a suboptimal estimator (i.e., $\alpha=-1$ fixed for MLMC) under best-case conditions. 
If the adaptive algorithm performs competitively with MLMC-BEST,
it suggests the effect of pilot sampling is comparable to the inefficiency of using the suboptimal MLMC estimator while accessing oracle covariance. 

The results are summarized in Figure~\ref{fig:ADAPT_boxplots_all} 
and Table~\ref{tab:mono_summary}. The table reports the average termination point of the algorithm over 20 trials ($N_{\text{pilot}}^\ast$), the average estimator variance using the remaining budget, and the variance reduction ratio (VRR) statistics. The VRR figures are computed as the ratio of the single-fidelity MC variance to the realized ACV estimator variance after pilot sampling.

These results show that all three adaptive methods perform similarly to MLMC-BEST: slightly worse under very small budgets, but comparable at larger budgets. Notably, ADAPT-IW often performs best, even outperforming MLMC-BEST in some cases. As total budget increases, the negative effect of pilot sampling diminishes since sufficient budget remains for ACV evaluation. 
Also, the variation in the variance reduction is diminished at higher budgets since more pilot samples are afforded and thus the covariance is estimated more accurately.
Comparing across methods, ADAPT-GGMVN and ADAPT-GG generally terminate later than ADAPT-IW, leading to more reliable but less aggressively reduced variances. 
In contrast, ADAPT-IW contracts more quickly, prompting earlier termination and stronger variance reduction but at the expense of reliability.
The difference arises from the modeling assumptions: the $\gamma$-Gaussian prior converges more slowly due to nonlinearities between correlations and $\gamma$ parameters, whereas the IW prior yields faster contraction. 

\begin{figure}[htbp]
    \centering
    \includegraphics[width=0.8\linewidth]{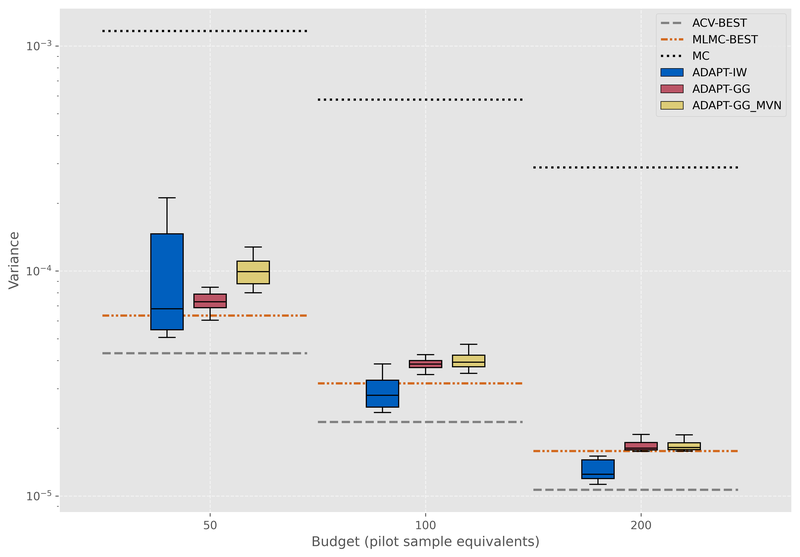}
    \caption{Results of the adaptive pilot sampling algorithm over 20 trials, under budgets equivalent to 50, 100, and 200 pilot samples. Performance is compared against ACV-BEST (gray), MLMC-BEST (orange), MC (black).}
    \label{fig:ADAPT_boxplots_all}
\end{figure}

\begin{table}[htbp]
\centering %
\renewcommand{\arraystretch}{1.1}
\begin{tabular}{|cccccc|}
\hline
\textbf{Method} & \textbf{Budget} & \textbf{\begin{tabular}[c]{@{}c@{}}Average \\ $N_{\mathrm{pilot}}^{\ast}$\end{tabular}} & \textbf{\begin{tabular}[c]{@{}c@{}}Average \\ Estimator \\ Variance\end{tabular}} & \textbf{\begin{tabular}[c]{@{}c@{}} VRR \\ Mean \end{tabular}} & \textbf{\begin{tabular}[c]{@{}c@{}}VRR \\ Standard\\ Deviation\end{tabular}} \\ \hline
\multirow{3}{*}{ACV-BEST} & 50 & - & $4.31 \times 10^{-5}$ & 27.1 & - \\
& 100 & - & $2.13 \times 10^{-5}$ & 27.1 & - \\
& 200 & - & $1.07 \times 10^{-5}$ & 27.0 & - \\ \hline
\multirow{3}{*}{MLMC-BEST} & 50 & - & $6.34 \times 10^{-5}$ & 18.2 & - \\
 & 100 & - & $3.17 \times 10^{-5}$ & 18.2 & - \\
 & 200 & - & $1.58 \times 10^{-5}$ & 18.4 & - \\ \hline
\multirow{3}{*}{ADAPT-IW} & 50 & 8 & $1.03 \times 10^{-4}$ & 14.7 & 6.54 \\
& 100 & 11.1 & $3.67 \times 10^{-5}$ & 18.9 & 5.62 \\
& 200 & 13.5 & $1.51 \times 10^{-5}$ & 21.1 & 4.80 \\ \hline
\multirow{3}{*}{ADAPT-GG} & 50 & 17.7 & $7.74 \times 10^{-5}$ & 15.4 & 2.46 \\
& 100 & 40.8 & $4.06 \times 10^{-5}$ & 14.6 & 1.97 \\
& 200 & 60.6 & $1.67 \times 10^{-5}$ & 17.3 & 0.87 \\ \hline
\multirow{3}{*}{ADAPT-GGMVN} & 50 & 19.4 & $1.04 \times 10^{-4}$ & 11.5 & 2.16 \\
& 100 & 37.6 & $4.16 \times 10^{-5}$ & 14.2 & 1.79 \\
& 200 & 61.3 & $1.67 \times 10^{-5}$ & 17.3 & 0.82 \\ \hline
\end{tabular}
\caption{Summary of adaptive pilot sampling results for the monomial benchmark, across budgets equivalent to 50, 100, and 200 pilot samples (20 trials each). }
\label{tab:mono_summary} 
\end{table}

Finally, we investigate the impact of prior informativeness. We set the budget to equivalent to $200$ pilot samples. 
Since ADAPT-GGMVN and ADAPT-GG behave similarly at this budget, we focus on ADAPT-GGMVN and ADAPT-IW. 
For ADAPT-GGMVN, we compare: 
\begin{itemize}
\item uninformative prior as described in Section~\ref{sss:c1_setup} but with all prior standard deviations inflated by a factor of 10; and

\item informative prior with $\boldsymbol{\gamma} \sim \mathcal{N}\left(\mu_{\gamma_{\mathrm{or}}},\mathrm{diag}([0.1,0.1,0.1,0.1])\right)$ and $\log \boldsymbol{\sigma}_{m} \sim \mathcal{N}(\log(\sigma_{m,\mathrm{or}}),0.1)$ for $m=0,\ldots,3$, where the means are the oracle values. 
\end{itemize}
For ADAPT-IW, we define an informative prior with $H_0$ corresponding to the oracle covariance and $\nu_0=20$, equivalent to 20 pseudo-samples. 
Unlike the $\gamma$-Gaussian case, the IW prior cannot be made less informative than $\nu_0=M+1$ to remain a proper distribution. 

Each configuration is run for $20$ trials (Figure~\ref{fig:boxplots_combined}).
Informative priors significantly improve both performance and reliability, as the posterior contracts with very few pilot samples. 
Thus, where users have accurate prior knowledge (e.g., from legacy runs or domain expertise), encoding this knowledge yields substantial benefits. 
In contrast, with little prior knowledge, more pilot samples are required and estimator performance will be lower. 
Interestingly, reducing informativeness below a certain threshold has little impact: the original and uninformative ADAPT-GGMVN cases produce nearly identical results, indicating that the baseline prior was already weakly informative. 
Overall, the algorithm adapts termination behavior to the informativeness of the prior, so users with strong prior knowledge can save considerable pilot-sampling resources.

\begin{figure}[H]
    \centering
    \subfloat[ADAPT-GGMVN \label{fig:boxplots_ggmvn}]{
        \includegraphics[width=0.45\linewidth]{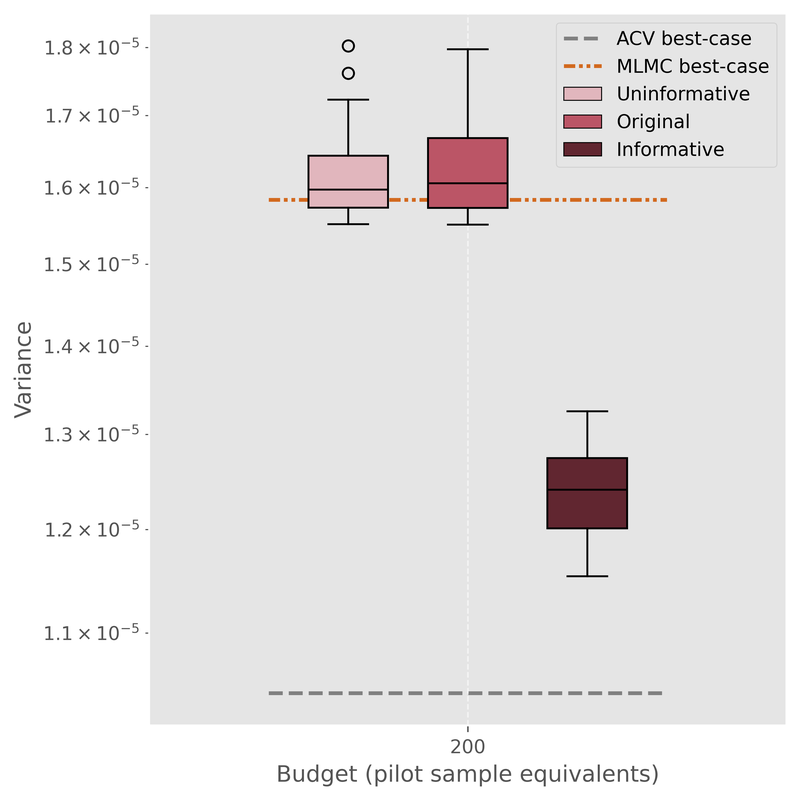}
    }
    \hfill
    \subfloat[ADAPT-IW \label{fig:boxplots_iw}]{
        \includegraphics[width=0.45\linewidth]{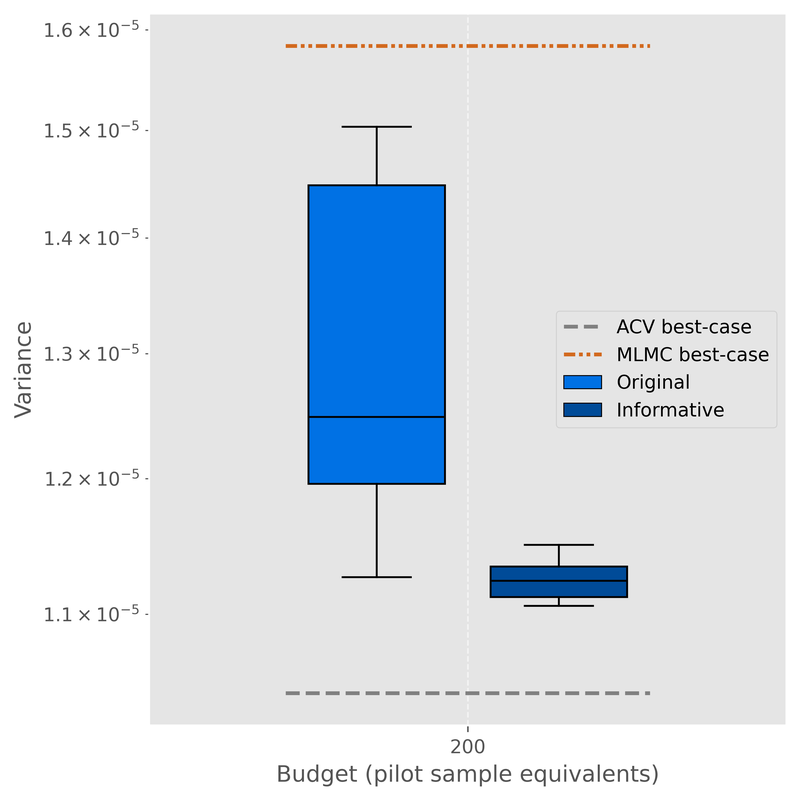}
    }
    \caption{Comparison of the informative and uninformative variants of each prior for a budget equivalent to 200 pilot samples for the monomial model ensemble.}
    \label{fig:boxplots_combined}
\end{figure}

\subsection{Case 2: Darcy flow PDE}

\subsubsection{Problem setup}

We next demonstrate our method on a more complex modeling example: a two-dimensional time-dependent Darcy flow problem, governed by a second-order parabolic PDE describing fluid flow through a porous medium. 
We follow the setup in PDEBench~\cite{takamoto_pdebench_2024}, 
a benchmark suite for scientific machine learning that includes PDE solvers and data-driven modeling tools.

In Darcy flow, the state variable $u(x)$ represents the flow pressure within a permeability field $a(x)$ over a unit-square domain $x\in[0,1]^2$. The dynamics are given by
\begin{equation}
    \partial_t u(x, t) - \nabla\left(a(x)\nabla u(x, t)\right) = g(x), \qquad x \in (0, 1)^2, \quad t>0, \label{eq: darcy_flow_2d_temporal}
\end{equation}
with initial condition $u(x,0)=a(x)$, homogeneous Neumann (zero-flux) boundary conditions, and a constant forcing term $g(x)=1$.

Uncertainty is introduced through the binary permeability field $\boldsymbol{a}(x)$, parameterized by $\boldsymbol{Z}=\{(\boldsymbol{\theta}_{x,i}, \boldsymbol{\theta}_{y,i}, \boldsymbol{\ell}_i)\}_{i=1}^{10}$, 
where each parameter is independently distributed as $\boldsymbol{\theta}_{x,i}\sim\mathcal{U}(0,1)$, $\boldsymbol{\theta}_{y,i}\sim\mathcal{U}(0,1)$, and $\boldsymbol{\ell}_i\sim\mathcal{U}(0,0.5)$. 
Samples of $\boldsymbol{Z}$ define Gaussian radial basis functions (RBFs), which are aggregated and then thresholded by their mean value to form binary maps with values $\{0.1,1\}$, yielding realizations of $\boldsymbol{a}(x)$.
The high dimensionality of $\boldsymbol{Z}$ motivates use of MC-based estimators.
An example permeability field is shown in Figure~\ref{fig:ic_raw_vs_final_darcy}.

\begin{figure}[htbp]
    \centering
    \includegraphics[width=0.75\linewidth]{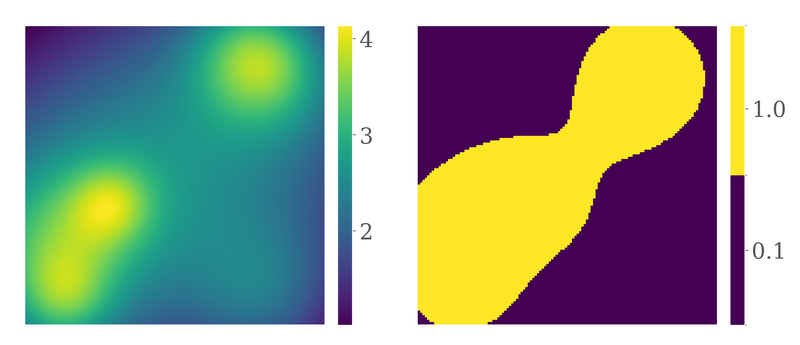}
    \caption{Construction of permeability field $a(x)$ in Darcy flow. The raw field (left) is thresholded at its mean value to produce the final binary map (right).}
    \label{fig:ic_raw_vs_final_darcy}
\end{figure}

Equation~\eqref{eq: darcy_flow_2d_temporal} is discretized using a second-order centered-difference scheme on a $128 \times 128$ uniform grid and integrated via Heun's method using $\Delta t = 0.03$ 
(set based on satisfying the Courant--Friedrichs--Lewy (CFL) condition for all runs) 
until a final time of $t=2$. At this horizon, the solution is effectively at steady state~\cite{takamoto_pdebench_2024}.
Representative realizations of $\boldsymbol{a}(x)$ and their corresponding pressure fields are shown in Figure~\ref{fig:ic_sol_darcy_10coords}, illustrating the strong dependence of the flow on the geometry of the permeability. 

\begin{figure}[htbp]
    \centering
    \includegraphics[width=0.95\linewidth]{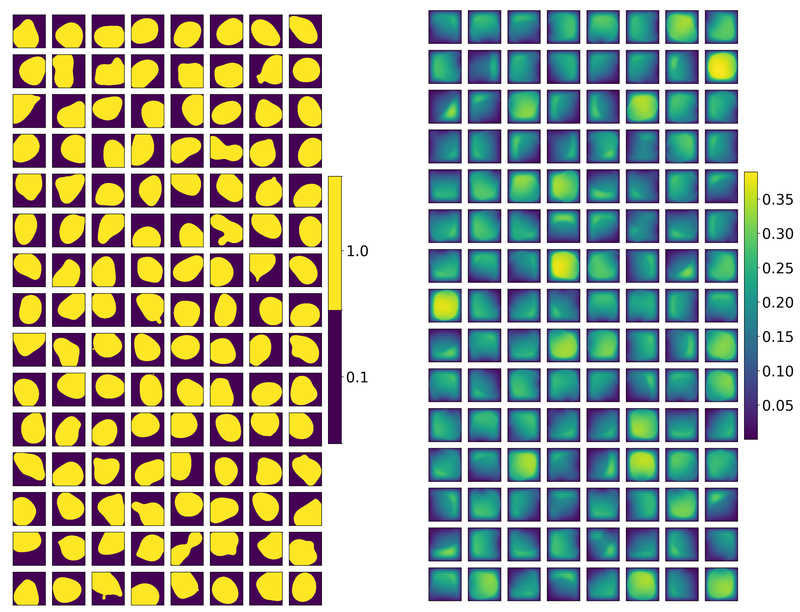}
    \caption{Example realizations of the binary permeability field (left) and their corresponding flow pressure solutions at $t=2$ (right). A total of 120 samples are generated for training the lower-fidelity models.}
    \label{fig:ic_sol_darcy_10coords}
\end{figure}

Finally, we define a scalar quantity of interest (QoI) to be the {relative} mean flow pressure at $t=2$ within a bounding box spanning $(0.6, 0.75)$ to $(0.8, 0.95)$ in the domain:
\begin{align}
\boldsymbol{Q} = f_0(\boldsymbol{Z})=\boldsymbol{\bar{u}}_{\mathrm{box}}(\boldsymbol{Z}, t=2) - \boldsymbol{\bar{u}}(\boldsymbol{Z}, t=2),
\end{align}
where $\boldsymbol{\bar{u}}_{\mathrm{box}}$ is the mean pressure within the box and $\boldsymbol{\bar{u}}$ is the mean pressure over the entire domain. The bounding box is illustrated in Figure~\ref{fig:sol_comparison_pde_unet_fno}.

\subsubsection{Model ensemble}

We employ two neural-network-based surrogates, also benchmarked in  \cite{takamoto_pdebench_2024}, as low-fidelity approximations of the high-fidelity PDE solver ($f_0$):
a U-net convolutional neural network ($f_1$), and a Fourier Neural Operator (FNO) ($f_2$). Both models are trained using high-fidelity model data (Figure~\ref{fig:ic_sol_darcy_10coords}), 
taking permeability fields $a(x)$ as input and predicting flow pressure fields $u(x,t=2)$.

The U-net ($f_1$) is a fully convolutional encoder-decoder network originally designed for biomedical image segmentation~\cite{ronneberger_u-net_2015}.
The encoder contracts the inputs $a(x)$ into a latent bottleneck representation, and the decoder progressively upsamples while incorporating skip connections from the encoding layers.  
This allows fine-scale features to be preserved during reconstruction, making U-nets well-suited for PDE input-output mappings.

The FNO ($f_2$), in contrast, first lifts the input $a(x)$ into a higher-dimensional representation 
through a shallow fully connected layer. 
A sequence of updates is then applied via compositions of a kernel integral operator and nonlinear activations, 
followed by projection to the original space via additional fully connected layers.
Assuming a stationary kernel, the convolutional theorem enables efficient parameterization in the Fourier space. 
The uniform discretization of the domain allows use of fast Fourier transform (FFT), and coupled with truncation of the Fourier series, this reduces the computational complexity for the forward pass. 
The neural operator framework, described in greater detail in \cite{Li2020}, has been successfully extended with physical priors (e.g., spherical Fourier transforms) for applications such as weather forecasting~\cite{bonev_spherical_2023,kurth_fourcastnet_2023}.

While deep neural networks typically require large training datasets, multi-fidelity estimation has more modest requirements: low-fidelity models need only to be sufficiently \textit{correlated} with high-fidelity outputs in the region of interest. 
Thus, even models trained on smaller datasets can provide substantial variance-reduction over single-fidelity MC.
In this case, both low-fidelity models are trained on 120 pairs of $(a(x), u(x,t=2))$ from Figure~\ref{fig:ic_sol_darcy_10coords}.  
The training loss is the mean absolute error (MAE) between predicted and ground-truth solution. 
Learning is carried out for 40 epochs, with the learning rate halving every 5 epochs.

To assess accuracy, the trained models are tested on 30 additional high-fidelity samples, with representative comparisons shown in Figure~\ref{fig:sol_comparison_pde_unet_fno}.
The U-net exhibits minor artifacts near permeability discontinuities but achieves better correlation with the high-fidelity QoI. 
The FNO, despite having an order of magnitude fewer parameters than the U-net, produces smoother fields but tends to under-predict the QoI. 
Timing and accuracy results are summarized in Table~\ref{tab: surrogate_summary}. 

\begin{figure}[htbp]
    \centering
    \includegraphics[width=0.8\linewidth]{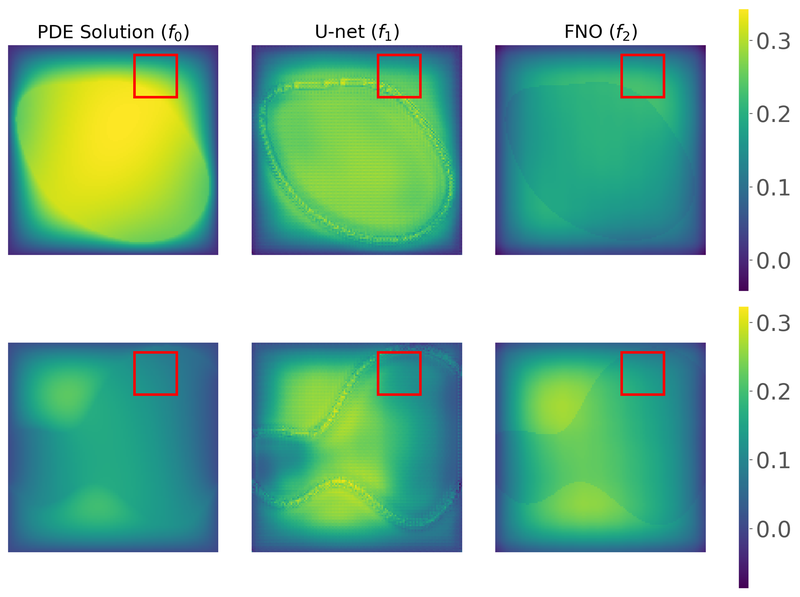}
    \caption{Flow pressure fields for two realization of $\boldsymbol{Z}$. In each row, solution generated by the high-fidelity solver (left) are compared with predictions from the trained low-fidelity models (middle and right). The red bounding box marks the region used to compute the scalar QoI, $Q_m=f_m(Z)$.}
    \label{fig:sol_comparison_pde_unet_fno}
\end{figure}

\begin{table}[htbp]
\centering
\begin{tabular}{ccc}
\toprule
\textbf{Model} & \textbf{Parameters} 
& \textbf{Batch run time [s]} \\ \hline
PDE solver ($f_0$)     & -                   
& 508                   \\ 
U-net ($f_1$)          & 7,762,465             
& 3.0                     \\
FNO ($f_2$)           & 465,377              
& 2.2                     \\ \bottomrule
\end{tabular}
\caption{Model sizes and run times for a batch of 30 test samples.}\label{tab: surrogate_summary}
\end{table}

Based on the observed runtimes, we normalize the cost vector as $w=[1, 0.006, 0.004]$. Training costs are not included since the low-fidelity models are fixed post-training and treated as black-box models during pilot sampling.
Future work will investigate trade-offs between allocating budget for model training versus pilot sampling for covariance estimation.

We specify the following relatively informative prior baseline.  For IW (ADAPT-IW), the prior correlation is
\begin{align}
    R_0 &= 
    \begin{bmatrix}
    1     & 0.8 & 0.7 \\
    0.8  & 1  & 0.7  \\
    0.7 & 0.7  & 1 \\     
    \end{bmatrix},\label{eqn:prior_mean_corr2}
\end{align}
with prior standard deviations
\begin{align}
    {\sigma}_0=\left[ 0.05, 0.05, 0.05 \right]^{\top},
\end{align}
yielding the prior covariance $\Sigma_0 = \mathrm{diag}({\sigma}_0) \, R_0 \, \mathrm{diag}({\sigma}_0)$. 
We assign a stronger prior correlation (0.8) between the high-fidelity model and U-net than between the high-fidelity model and FNO (0.7), consistent with their observed prediction quality.
We also set $\nu_0 = 8$ and define $H_0$ using \eqref{eqn:iw_mean}, giving the prior distribution $\mathrm{IW}(H_0,\nu_0)$.

As in the monomial case, for $\gamma$-Gaussian (ADAPT-GG and ADAPT-GGMVN), we map $R_0$ into the $\gamma$ space:
\begin{align}
\mu_{\gamma,0} = [0.9429, 0.6573, 0.6573]^{\top},\label{eq:mu_gamma_darcy}
\end{align}
with unit variances. The prior for the standard deviation is specified as $\log(\boldsymbol{\sigma}_{m})\sim \mathcal{N}(0.05,1)$ for $m=0,\ldots,2$ .

\subsubsection{Results}

For the Darcy flow problem, we set a total budget of $B_\mathrm{tot}=200 \sum_m w_m$, with 4 initial pilot samples, batch size $k=2$, and projection horizon $n_{\mathrm{steps}}=6$ for ADAPT-GGMVN, and $n_{\mathrm{steps}}=3$ for ADAPT-IW and ADAPT-GG. Results are reported over 5 trials, each using 800 MC samples to evaluate expected and projected losses. 

Table~\ref{tab:pdebench_summary} summarizes the outcomes along with VRR statistics relative to the single-fidelity baseline, in the same format as Table~\ref{tab:mono_summary}. 
Unlike the monomial benchmark, repeated high-fidelity evaluations are too costly to obtain exact oracle covariances. 
Instead, oracle correlations are approximated from 200 high-fidelity and low-fidelity model evaluations, while the low-fidelity standard deviations can be refined at very low cost by evaluating an additional 1000 samples.

Figure~\ref{fig:df_curves} shows the projected and actual expected loss curves for the ADAPT-GGMVN method. 
The expected loss plateaus between 10--14 pilot samples, consistent with the estimated average $N_{\mathrm{pilot}}^{\ast}$ across the three adaptive termination methods. 
All methods achieve significant variance reduction despite the limited total budget.
Among the approaches, ADAPT-GG achieves the best performance, with a mean VRR of 4.96. For ADAPT-IW and ADAPT-GGMVN, using the posterior mean of $\boldsymbol{\Sigma}$ as the point estimate to select the ACV hyperparameters improves the average estimator variance by 8\% and 10\% respectively compared to the result based on the unbiased sample covariance.
Overall, in addition to enabling adaptive pilot sampling, Bayesian inference provides a framework to inform the covariance estimation process with prior information from the model training process, leading to well-calibrated ACV hyperparameters.
The relatively small pilot sampling resources required to achieve significant variance reduction is also noteworthy---this has promising implications for modelers who wish to produce data-driven low-fidelity models specifically for the multi-fidelity UQ task but also keeping offline costs for model training and pilot sampling to a minimum.

\begin{figure}[htbp]
    \centering
    \includegraphics[width=\linewidth]{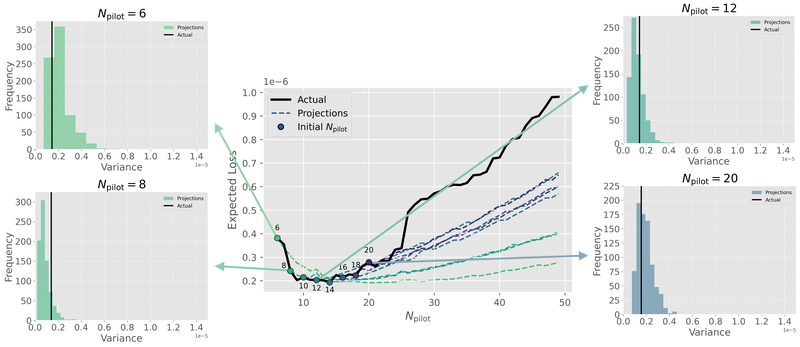}
    \caption{Projected and actual expected loss curves for the ADAPT-GGMVN method under a single random seed with total budget equivalent to $200$ pilot samples. Histograms of estimated variances are compared against the true variance at $N_{\mathrm{pilot}}=6$, 8, 12 and 20. As additional samples are collected, the histograms concentrate toward the true value.
    }
    \label{fig:df_curves}
\end{figure}

\begin{table}[htbp]
\centering
\renewcommand{\arraystretch}{1.1}
\begin{tabular}{|ccccc|}
\hline
\textbf{Method} & \textbf{\begin{tabular}[c]{@{}c@{}} Average \\ $N_{\mathrm{pilot}}^{\ast}$\end{tabular}}  & \textbf{\begin{tabular}[c]{@{}c@{}} Average \\ Estimator \\ Variance \end{tabular}}  & \textbf{\begin{tabular}[c]{@{}c@{}} VRR \\ Mean \end{tabular}} & \textbf{\begin{tabular}[c]{@{}c@{}}VRR \\ Standard\\ Deviation\end{tabular}} \\ \hline
ADAPT-IW & 10.8 & $1.77 \times 10^{-6}$ & 4.83 & 0.49\\
\hline
ADAPT-GG & 12.8 & $1.723 \times 10^{-6}$ & 4.96 & 0.55\\
\hline
ADAPT-GGMVN & 10.4 & $2.15 \times 10^{-6}$ & 4.06 &  0.79\\
\hline
\end{tabular}
\caption{Summary of adaptive pilot sampling results for the Darcy flow problem, across budgets equivalent to 200 pilot samples (5 trials). In contrast, ACV-BEST achieves an estimator variance of $1.3 \times 10^{-6}$ for a VRR of 6.5.}
\label{tab:pdebench_summary} 
\end{table}

\section{Conclusions}
\label{sec:conclusions}

In this work, we advanced sampling-based multi-fidelity UQ by explicitly quantifying uncertainty in the model-output covariance matrix. 
A central insight is that inaccurate empirical covariance estimates---common but largely ignored in existing approaches---can significantly degrade estimator performance.
Unlike traditional methods, we address this issue through probabilistic modeling: treating the covariance matrix as an uncertain parameter and updating beliefs via Bayesian inference conditioned on pilot samples.  
This perspective naturally connects covariance uncertainty to estimator inefficiency, motivating the development of an interpretable expected-loss criterion---built on a new $\gamma$-Gaussian projected-posterior approach---for pilot-sampling termination under finite-sample budgets.

Our results demonstrate that 
Bayesian inference significantly improves multi-fidelity estimation under limited pilot budgets by regularizing noisy covariances estimates and incorporating prior knowledge into the inference process. 
This framework not only supports adaptive pilot-sampling termination but also provides UQ for the estimator variance itself, avoiding misleading overconfidence. 
    The expected-loss criterion further decomposes inefficiency into accuracy and cost components, offering an interpretable balance between exploration and exploitation. 
    Across both benchmark and PDE-based problems, the method performs competitively with baselines that assume oracle covariance knowledge.

Some limitations remain. 
The results depend on modeling assumptions, including priors, likelihoods, and posterior-projection schemes, which require careful specification.
Algorithm hyperparameters such as projection horizon and number of MC samples must also be tuned to avoid premature or delayed termination. 
While this flexibility is ultimately a strength---allowing users to tailor assumptions to problem-specific structure---it does introduce overhead.
In addition, MC-based loss estimation, though embarrassingly parallel, is computationally heavier than closed-form variance expressions and may be less attractive when model evaluations themselves are inexpensive.

Several directions for future work appear promising. Risk-aware termination criteria (e.g., quantiles, conditional value at risk) would better serve applications sensitive to adverse outcomes. 
For ensembles including lower-fidelity models trained from high-fidelity data, objectives that co-allocate resources across training, pilot sampling, and final estimation could improve overall efficiency. 
Finally, developing methods to update covariance beliefs from partial ensemble evaluations, rather than the full set of models at each pilot sample, could yield further computational savings by enabling non-uniform pilot sampling.

\appendix
\section{Bijective Parameterization of Correlation Matrices} 
\label{app: bij_param}

\begin{algorithm}[htbp]
\begin{algorithmic}[1]
\caption{Forward $\gamma$-transform.}\label{alg:gamma-forward}
\Require PSD correlation matrix $R\in\RR^{M\times M}$
\Ensure correlation parameters $\gamma \in \RR^{\frac{M(M-1)}{2}}$
\State 
$R = Q \Lambda Q'$
\State 
$\log R = Q \log (\Lambda) Q'$
\State 
$\gamma = \text{vecl}\left(\log(R)\right)$
\end{algorithmic}
\end{algorithm}

\begin{algorithm}[htbp]
\begin{algorithmic}[1]
\caption{Inverse $\gamma$-transform (fixed-point iteration).}\label{alg:gamma-inverse}
\Require correlation parameters $\gamma \in \RR^{\frac{M(M-1)}{2}}$, 
tolerance for fixed-point iteration $\texttt{tol}$
\Ensure PSD correlation matrix $R\in\RR^{M\times M}$
\State 
$A \gets \mathbf{0}_{n \times n}$
\State Fill lower triangle of $A$ (excluding diagonal) with entries of $\gamma$
\State 
$A \gets A + A^\top$
\State 
$d \gets \sqrt{n}$, $i \gets -1$
\While{$d > (\sqrt{M} \, \texttt{tol})$}
    \State $\Delta \gets \log(\mathrm{diag}(\exp(A)))$
    \State $\mathrm{diag}(A) \gets \mathrm{diag}(A) - \Delta$
    \State $d \gets \| \Delta \|$
    \State $i \gets i+1$
\EndWhile
\State $R \gets \exp(A)$
\State 
$\mathrm{diag}(R) \gets \mathbf{1}$
\end{algorithmic}
\end{algorithm}

\section*{Acknowledgments}

This material is based upon work supported in part by the National Science Foundation Graduate Research Fellowship under Grant No. DGE 1841052, and the Department of Navy award N00014-23-1-2735 issued by the Office of Naval Research.
This research is also supported in part through computational resources and services provided by Advanced Research Computing at the University of Michigan, Ann Arbor.

\bibliographystyle{unsrtnat}
\bibliography{references_manuscript}

\end{document}